\documentclass[12pt]{article}
\usepackage[utf8]{inputenc}
\usepackage[a4paper]{geometry}
\usepackage{url}
\usepackage{color}
\usepackage{array}
\usepackage{textcomp}
\usepackage{multirow}
\usepackage{amsmath}
\usepackage{amsthm}
\usepackage{graphicx}
\usepackage{setspace}
\usepackage[authoryear,round]{natbib}
\usepackage[unicode=true,
 bookmarks=false,
 breaklinks=false,pdfborder={0 0 0},pdfborderstyle={},backref=false,colorlinks=true]
 {hyperref}
\hypersetup{
 citecolor=blue, linkcolor=blue, urlcolor=blue}

\makeatletter

\newcommand{\lyxmathsym}[1]{\ifmmode\begingroup\def\b@ld{bold}
  \text{\ifx\math@version\b@ld\bfseries\fi#1}\endgroup\else#1\fi}

\providecommand{\tabularnewline}{\\}

\numberwithin{equation}{section}
\theoremstyle{definition}
 \newtheorem{example}{\protect\examplename}
\theoremstyle{remark}
\newtheorem{rem}{\protect\remarkname}
\theoremstyle{definition}
\newtheorem{defn}{\protect\definitionname}
\theoremstyle{plain}
\newtheorem{fact}{\protect\factname}
\theoremstyle{remark}

\theoremstyle{plain}
\newtheorem{cor}{\protect\corollaryname}
\theoremstyle{plain}
\newtheorem{lem}{\protect\lemmaname}
\theoremstyle{plain}
\newtheorem{prop}{\protect\propositionname}
\theoremstyle{plain}
\newtheorem*{cor*}{\protect\corollaryname}
\theoremstyle{plain}
\newtheorem{thm}{\protect\theoremname}
\theoremstyle{definition}
\newtheorem*{example*}{\protect\examplename}
\theoremstyle{plain}
\newtheorem*{prop*}{\protect\propositionname}
\newcounter{primetheorem}
\renewcommand{\theprimetheorem}{\arabic{primetheorem}'}

\newenvironment{primedtheorem}[1][]{%
    \refstepcounter{primetheorem}%
    \noindent\textbf{Theorem \theprimetheorem.} #1}{}

\usepackage[T1]{fontenc}
\usepackage{lmodern}
\definecolor{green}{RGB}{00, 180, 00}
\definecolor{red}{RGB}{180, 00, 00}
\usepackage{amsfonts}
\usepackage{changepage}
\makeatother

\providecommand{\claimname}{Claim}
\providecommand{\corollaryname}{Corollary}
\providecommand{\lemmaname}{Lemma}
\providecommand{\definitionname}{Definition}
\providecommand{\examplename}{Example}
\providecommand{\factname}{Fact}
\providecommand{\propositionname}{Proposition}
\providecommand{\remarkname}{Remark}
\providecommand{\theoremname}{Theorem}

\begin{document}
\title{Heterogeneous Noise and Stable Miscoordination}
\author{Srinivas Arigapudi\thanks{Department of Economic Sciences, IIT Kanpur, \protect\protect\href{mailto:arigapudi@iitk.ac.in}{arigapudi@iitk.ac.in}.}
\and Yuval Heller\thanks{Department of Economics, Bar-Ilan University, \protect\protect\href{mailto:yuval.heller@biu.ac.il}{yuval.heller@biu.ac.il}.} \and Amnon Schreiber\thanks{Department of Economics, Bar-Ilan University, \protect\protect\href{mailto:amnon.schreiber@biu.ac.il}{amnon.schreiber@biu.ac.il}.}\thanks{This paper replaces an 
earlier working paper titled \textquotedblleft Sampling
Dynamics and Stable Mixing in Hawk--Dove Games.\textquotedblright \ We are grateful 
to Vince Crawford, Josef Hofbauer, Bill Zame, and seminar and conference participants at Bar-Ilan University, University of Haifa, Learning, Evolution and Games, SWET 2023 (UC Irvine), UC Los Angeles, UC Riverside, UC San Diego, and University of Southern California, 
as well as the editor and anonymous reviewers, for their  valuable feedback. We also thank Luis Izquierdo and Segismundo Izquierdo for their help in implementing some of the simulations carried out in the paper. YH gratefully acknowledges the financial support of the European Research
Council (\#677057), the Israeli Science Foundation (\#448/22), and the US--Israel Binational Science Foundation (\#2020022). SA gratefully acknowledges the financial support of the Fine fellowship of the Technion and an Initiation grant from IIT Kanpur.}}
\maketitle
\begin{abstract}
\noindent 

Coordination games feature two types of equilibria: pure equilibria, where players successfully coordinate their actions, and mixed equilibria, where players frequently experience miscoordination. We investigate learning dynamics where agents observe the actions of a random sample of their opponents. First, we show that when all agents have the same sample size, whether it is small or large, their behavior converges to one of the pure coordinated equilibria. By contrast, our main results show that stable miscoordination often persists when some agents make decisions based on small samples while others rely on large samples. 

Final pre-print of a manuscript accepted for publication in \textit{American Economic Journal: Microeconomics}.
\end{abstract}

\section{Introduction\label{sec:Introduction}}
Coordination games provide a valuable framework for modeling real-life situations, where the optimal response to an opponent's action is often to mimic that action. In the context of two-player two-action coordination games, we observe three equilibria of two distinct types: two strict pure equilibria, allowing successful coordination of actions, and a mixed equilibrium, characterized by frequent miscoordination.

A fundamental result in evolutionary
game theory is that under a broad set of learning dynamics, the mixed
equilibrium is unstable and populations, in which agents are randomly
matched to play coordination games,  converge to everyone playing
one of the pure equilibria (as surveyed in Section \ref{Sec-discussion}).\footnote{\label{footnote1}
For example, \citet[p. 946]{myerson2015tenable} state that ``if such a [coordination] game is often played in culturally familiar settings, the
mixed equilibrium appears very unlikely. One would expect individuals to develop an understanding that coordinates their expectations at one of the strict equilibria.''}
By contrast, in this paper we show that miscoordination 
can be stable if the populations are heterogeneous in the sense that some agents rely on anecdotal evidence, while other agents have accurate information about the opponents' behavior. 

\paragraph{Highlights of the Model}

Consider a setup in which pairs of agents from an infinite  population are repeatedly randomly matched to play a (one-shot) two-action coordination
game. Agents occasionally die and are replaced by new agents 
(or,
alternatively, agents occasionally receive opportunities to revise
their actions). The new agents do not have precise information about
the aggregate behavior in the population, and estimate
this behavior through sampling. Specifically, the population
is characterized by a distribution of sample sizes.
 Each agent with sample size $k$ observes the behavior of $k$ random
opponents, and then adopts the action that is a best response to her
sample (with an arbitrary tie-breaking 
rule).\footnote{This behavior can be interpreted by each new agent utilizing her own sample to calculate a maximum likelihood estimate of the aggregate behavior.} 
These learning dynamics,
which seem plausible in various setups, are called \emph{sampling
best-response dynamics }(\citealp*{sandholm2001almost,
osborne2003sampling,oyama2015sampling},
henceforth abbreviated as \emph{sampling dynamics}).

\paragraph{Main Results}
We first show that heterogeneity in sample size  is
necessary for stable miscoordination. Specifically, 
Theorem \ref{thm1} shows that if
all agents have the same sample size, 
then all states with miscoordination are unstable.
Our next two results show that introducing heterogeneity in sample size often yields stable miscoordination. Theorem \ref{thm2} shows that for many coordination games, there are many heterogeneous distributions  that admit a stable interior state with miscoordination. Theorem \ref{thm3} further shows that under some assumptions on the payoffs, these stable interior states always induce a high level of miscoordination,  in the sense that most pairs of matched agents miscoordinate.

The intuition for why heterogeneity is necessary (Theorem \ref{thm1}) and often sufficient (Theorems \ref{thm2} and \ref{thm3}) 
for stable miscoordination is as follows. 
The probability that a new agent plays the first action $a$ is an increasing 
function of the share of opponents who play $a$, and the probability coincides with this share if it is either zero or one. 
The derivative of this
function captures the sensitivity of the behavior of new agents to small changes in the share of opponents who play $a$.

Observe that the average value of the derivative in the unit interval is $\frac{1-0}{1-0}=1$. One can show that if all agents have the same sample size, the slope is a unimodal function of the share of agents who play $a$ and the interior stationary state is close to the peak of the slope. 
This implies that the slope in the interior stationary state is greater than one, which, in turn, implies that small perturbations gradually increase: an initial small increase in the share of agents playing $a$ would eventually take the population to the state in which everyone plays $a$. 

By contrast, for heterogeneous populations in which some agents have small samples and the others have large samples, the slope of the probability of playing $a$ as a function of the share of agents who play $a$ is bimodal, and in many cases the interior stationary state is close to the local minimum between the two peaks. This implies that the slope in the interior equilibrium is less than one, which, in turn, implies that small perturbations vanish, and that the interior stationary state is locally stable.

Our final main result (Theorem \ref{thm:global-mixed}) fully characterizes  environments in which the populations converge
to miscoordination from almost all initial states (rather than being only locally stable as in Theorems \ref{thm2} and \ref{thm3}). This
holds for a relatively narrow set of environments in which the coordination game is asymmetric, the two players have different preferred outcomes, and 
a sufficiently large share of agents in the population have sample size 1.

\paragraph{Extensions and Insights}

Our baseline model analyzes two-player two-action coordination games with a specific class of evolutionary dynamics, namely, sampling dynamics. In Section \ref{Sec:Extensions}, we relax various key assumptions (detailed in Online Appendix \ref{online-appendix}) and demonstrate the robustness of our main results. Specifically, we extend our results to games with more than two actions in Appendix \ref{sec-multiple-actions}, and to games with more than two players in Appendix \ref{sec-multiple-players}. Lastly, Appendix \ref{sec-logit} numerically demonstrates that our qualitative result holds across a broad class of learning dynamics, including the widely used logit dynamics (\citealp{fudenberg1995consistency}).

Taken together, our results show that the conventional wisdom that
miscoordination is unstable is not accurate. Miscoordination can be
stable in heterogeneous populations in which some agents rely on limited data, while other agents have access to comprehensive
data. We discuss the empirical relevance of our results in Section \ref{subsec-empirical}, and their  experimentally testable implications in Section \ref{subsec-experimental}.

The asymptotically stable  interior  states identified in our results  are typically not Nash equilibria. A significant share of agents with small samples choose actions that are not best responses to the true aggregate behavior. This prompts the question of why these agents do not learn to increase their sample sizes, or why they are not eliminated from the population. We propose two explanations. First, many real-life scenarios involve both inexperienced and experienced players. For instance, in the housing market, there are first-time buyers or sellers who interact with seasoned real-estate investors (coordination games in housing markets are discussed in Section \ref{subsec-empirical}). The former group has limited data, and acquiring more reliable data may be costly or challenging, especially given the inherent biases of real estate agents. As a steady stream of first-time participants enters the market, 
the share of agents with small samples remains stable. Another explanation is the ``law of small numbers,'' a common cognitive  bias where individuals mistakenly believe that small samples accurately reflect the larger population (\citealp*{tversky1971belief}). This bias leads players with small samples to underinvest in obtaining more data.

\paragraph{Structure}

Section \ref{sec:symmetric} analyzes symmetric coordination games. In Section \ref{sec:general}, we extend the analysis to general coordination games. In Section \ref{Sec-discussion}, we discuss our results, and in Section \ref{sec:related}
 we survey the related literature. We conclude in Section \ref{sec:Conclusion}.  
The formal proofs are presented in Appendix \ref{appendix-proofs}.  Online Appendix \ref{online-appendix} presents extensions of our model.

\section{Symmetric Coordination Games}\label{sec:symmetric}

\subsection{Baseline Model}
Consider two-player two-action symmetric coordination game in normal form where each player can choose between actions $A = \{a, b\}$. The standard one-parameter payoff matrix for a coordination game is shown on the left side of Table \ref{tab-symmetric}: players receive a low payoff (normalized to 0) if they miscoordinate (i.e., one player chooses $a$ while the opponent chooses $b$). If they coordinate on both choosing action $b$, they receive a high payoff (normalized to 1). If they coordinate on both choosing action $a$, they receive a payoff of $u > 0$.

\begin{table}

\caption{Two-Player Two-Action Symmetric Coordination Game}\label{tab-symmetric}

\centering{}%
\begin{tabular}{c|c|c|c|}
\multicolumn{2}{c}{} & \multicolumn{2}{c}{\textcolor{red}{Opponent}}\tabularnewline
\cline{3-4} \cline{4-4} 
\multicolumn{2}{c|}{} & \textcolor{red}{~~$a$~~} & \textcolor{red}{~~$b$~~}\tabularnewline
\cline{2-4} \cline{3-4} \cline{4-4} 
\multirow{2}{*}{\textcolor{blue}{Player}} & \textcolor{blue}{$a$} & \textcolor{blue}{$u$} & \textcolor{blue}{$0$}\tabularnewline
\cline{2-4} \cline{3-4} \cline{4-4} 
 & \textcolor{blue}{$b$} & \textcolor{blue}{0} & \textcolor{blue}{$1$}\tabularnewline
\cline{2-4} \cline{3-4} \cline{4-4} 
\multicolumn{4}{c}{Standard representation}\tabularnewline
\multicolumn{4}{c}{$\frac{u_{11}-u_{21}}{u_{22}-u_{12}}>0$}\tabularnewline
\end{tabular}~~~~~~~~~~~~~~%
\begin{tabular}{c|c|c|c|}
\multicolumn{2}{c}{} & \multicolumn{2}{c}{\textcolor{red}{Opponent}}\tabularnewline
\cline{3-4} \cline{4-4} 
\multicolumn{2}{c|}{} & \textcolor{red}{~~$a$~~} & \textcolor{red}{~~$b$~~}\tabularnewline
\cline{2-4} \cline{3-4} \cline{4-4} 
\multirow{2}{*}{\textcolor{blue}{Player}} & \textcolor{blue}{$a$} & \textcolor{blue}{$u_{11}$} & \textcolor{blue}{$u_{12}$}\tabularnewline
\cline{2-4} \cline{3-4} \cline{4-4} 
 & \textcolor{blue}{$b$} & \textcolor{blue}{$u_{21}$} & \textcolor{blue}{$u_{22}$}\tabularnewline
\cline{2-4} \cline{3-4} \cline{4-4} 
\multicolumn{4}{c}{Original representation}\tabularnewline
\multicolumn{4}{c}{$u_{11}>u_{21},\,\,u_{22}>u_{12}$}\tabularnewline
\end{tabular}
\end{table}

In Appendix \ref{subsec-general-coord}, we formally show that this standard one-parameter representation captures, without loss of generality, all two-action symmetric coordination games, 
as represented by the four-parameter form on the right side of Table \ref{tab-symmetric}. This is because our dynamics, as defined in Eqs. \eqref{eq:dynamics-sym} and \eqref{eq:action-sampling-dyanmics}, depend only on the differences between a player's payoffs in action profiles where the opponent plays the same action. Consequently, the dynamics are invariant to (1) adding a constant to both payoffs of the row player within the same column, and (2) dividing all of a player's payoffs by a positive constant.
 
Let \textbf{boldfaced} letters denote profiles (vectors of length two); for example, we let $\textbf{a}$ denote the action profile $(a, a)$. We identify each mixed action by the probability it assigns to the first action ($a$), denoted by $p \in [0, 1]$. The degenerate mixed action $1$ (resp., $0$) is identified with the pure action $a$ (resp., $b$). Note that the coordination game admits three Nash equilibria, all of which are symmetric: two pure equilibria, $\textbf{a}$ and $\textbf{b}$, where players coordinate their actions, and a mixed equilibrium $\boldsymbol{p}^{NE} = \left(\frac{1}{1+u}, \frac{1}{1+u}\right)$, where players frequently experience miscoordination.

\subsection{Evolutionary Dynamics}\label{subsec-one-population-dynamic}
Consider a continuum of agents with total mass one, who are randomly matched to play the one-shot coordination game. Aggregate behavior at time $t\in\mathbb{R}^{+}$ is described by
a \emph{state} $p\in\left[0,1\right]$, representing the share of agents playing action $a$ at time $t$. A state $p$
is \emph{interior} (i.e., mixed) if $p\in(0,1).$ 

Agents die at a constant rate of 1, and are replaced by new agents.
The evolutionary process is represented by a function $w:\left[0,1\right]\rightarrow\left[0,1\right]$,
which describes the share of new agents who
play action $a$ as a function of the current state. Thus, the
instantaneous change in the share of agents who
play $a$ is given by the
following dynamics (where $w$ 
 is as defined in Eq. \ref{eq:action-sampling-dyanmics} below):
 \vspace{-7pt}
 \begin{equation}\label{eq:dynamics-sym}
\dot{p}=w(p)-p.
 \vspace{-7pt}
\end{equation}

\paragraph{Sample Sizes}

We allow heterogeneity in sample size used by new agents. Let
$\theta\in\Delta\left(\mathbb{N}_{>0}\right)$ denote the distribution
of sample sizes of new agents. 
For simplicity, we assume that $\theta$
has a finite support. A share $\theta(k)$ of the
new agents have a sample of size $k$.  
Let $\text{supp}(\theta)$
denote the support of $\theta$. If there exists some $k$ for
which $\theta(k)=1,$ then we use $k$ to denote the degenerate
(homogeneous) distribution $\theta\equiv k$.
\begin{defn}
An \emph{environment} is a pair $(u,\theta)$,
where $\left(u\right)$ describes the payoff from coordinating on $a$ and $\theta$ describes
the distribution of sample sizes. 
\end{defn}

\paragraph{Sampling Best-Response Dynamics}

Sampling best-response dynamics (henceforth  \emph{sampling dynamics}; \citealp*{sandholm2001almost,oyama2015sampling}
) fit situations
in which agents do not know the exact distribution of actions in the
population. New agents estimate the aggregate behavior
by sampling opponents' actions. Specifically, each new agent with
sample size $k$ 
samples $k$ randomly
drawn agents from the population and then plays the action
that is the best response to the sample. 
One possible interpretation of this behavior is that a new agent utilizes her own sample to calculate a maximum likelihood estimation for the overall behavior of the population (or utilizes any other unbiased estimation procedure $\grave{\textrm{a}}$ la \citealp{salant2020statistical}).

To simplify notation,
we assume that in the case of a tie, the new agent plays $a$. All of our main
results 
(namely, Theorems \ref{thm1}--\ref{thm4}) remain the same under any tie-breaking rule.
Let $X(k,p)\sim Bin\left(k,p\right)$ denote a random variable
with binomial distribution with parameters $k$ (number of trials)
and $p$ (probability of success in each trial), which is interpreted
as the number of $a$'s in the sample. Observe that the sum of
payoffs from playing action $a$ against the sample is $u\cdot X(k,p)$
and the sum of payoffs from playing action $b$ against the sample
is $k-X(k,p)$.
This implies that action $a$ is a best response to a sample of size
$k$ iff $u\cdot X(k,p)\geq k-X(k,p)\Leftrightarrow X(k,p)\ge\frac{k}{u+1}$.
This, in turn, implies that the sampling dynamics for environment $\left(\boldsymbol{u},\boldsymbol{\theta}\right)$
are given by
\begin{equation}
w(p)=\sum_{k\in\text{supp}(\theta)}\theta(k)\cdot\Pr\left(X\left(k,p\right)\geq\frac{k}{u+1}\right),\label{eq:action-sampling-dyanmics}
\end{equation}
where $w(p)$ represents the probability that a random new agent (with a sample size distributed according to $\theta$) finds action $a$ to be the best response to her sample and thus chooses $a$).

Observe that $\Pr\left(X\left(k,p\right)\geq\frac{k}{u+1}\right)=\sum_{l=\left\lceil \frac{k}{u+1}\right\rceil }^{k}\left(\begin{array}{c}
k\\
l
\end{array}\right)p^{l}\left(1-p\right)^{k-l}$ , which is a polynomial of $p_{j}$ of degree $k$. This implies
that $w(p)$ is a polynomial of degree $\max(\textrm{supp}(\theta_{i}))$.

\subsection{Preliminary Analysis and Examples}
Stationary states (which have been termed \emph{sampling equilibrium} in \citealp{osborne2003sampling}) are states in which new agents play the same distribution of actions as the incumbents. Eq. (\ref{eq:dynamics-sym}) immediately implies that $\dot p>0$ iff $w(p)>p$, that $\dot p<0$ iff $w(p)<p$, and that $p$ is stationary iff $p$ is a fixed point of $w$.  This, in turn, implies the following:\footnote{See Appendix \ref{sub-standard-definitions} for the formal definitions of dynamic stability. 
A state $p$ is \emph{asymptotically stable} if populations starting close to $p$ stay close to $p$ and eventually converge to $p$. A state $p$ is \emph{unstable} if there exist populations starting arbitrarily close to $p$ that move away from $p$.}

\begin{fact}\label{fact-w-stable}
    
\begin{enumerate}
    \item State 0 is asymptotically stable iff $w(p)<p$ in a right neighborhood of 0.
    \item State 1 is asymptotically stable iff $w(p)>p$ in a left neighborhood of 1.
   \item An interior stationary state $p\in(0,1)$ is asymptotically stable iff $w(p)>p$ in a left neighborhood of $p$ and $w(p)<p$ in a right neighborhood of $p$. 
       \item A neighboring stationary state of an asymptotically stable state is unstable.     
    \item The limit $\lim_{t\rightarrow\infty}p\left(t\right)$ exists and is a stationary state for any initial state $p(0).$ 
\end{enumerate}
\end{fact}
Figure \ref{figure-symmetric} illustrates the phase plots of the sampling dynamics and the $w(p)$ curves for four environments, all of which have $u = 1.2$ but differ in their distributions of sample sizes. The top-left panel depicts the case in which all agents 
have sample size 1. In this scenario, $w(p)\equiv p;$ i.e., all states are stationary. By contrast, if at least some agents in the population sample multiple actions, then $w(p)-p$ is a polynomial of positive degree with a finite number of solutions. This implies that $w$ has a finite number of fixed points, and, therefore, the dynamics have a finite number of stationary states.
\begin{fact}\label{fact-fintie-sym}
If $\theta\equiv 1$, then all states are stationary. Otherwise, there exists a finite number of stationary states. 
\end{fact}

\begin{figure}[t]
\caption{Phase Plots and $w(p)$ Curves for Various Environments with
$u=1.2$}\label{figure-symmetric}

\begin{centering}
\begin{tabular}{cc}
\includegraphics[scale=0.37]{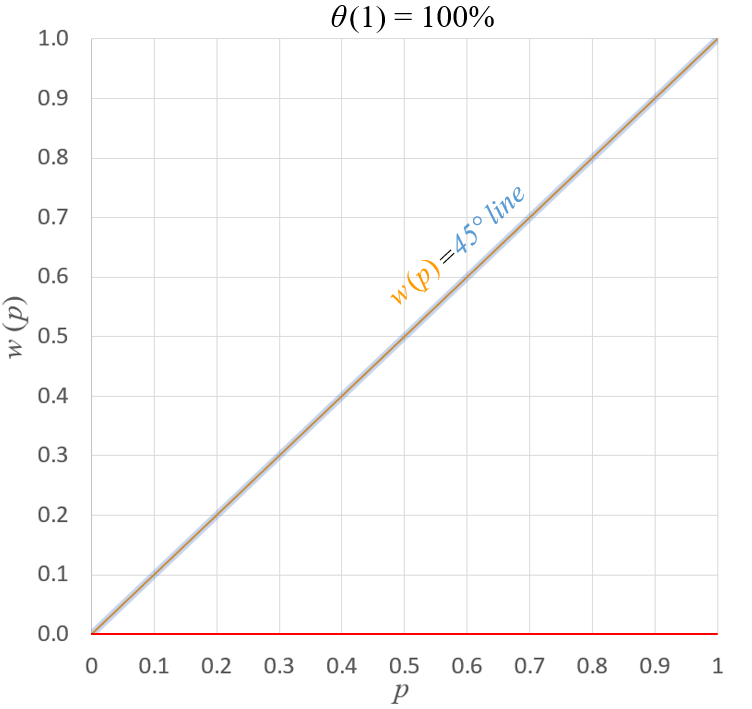} & \includegraphics[scale=0.37]{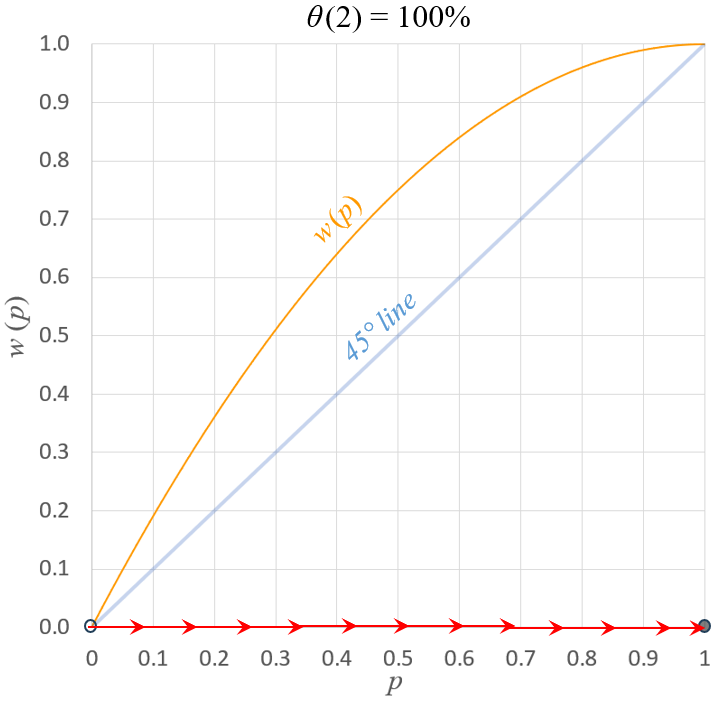}\tabularnewline
\includegraphics[scale=0.37]{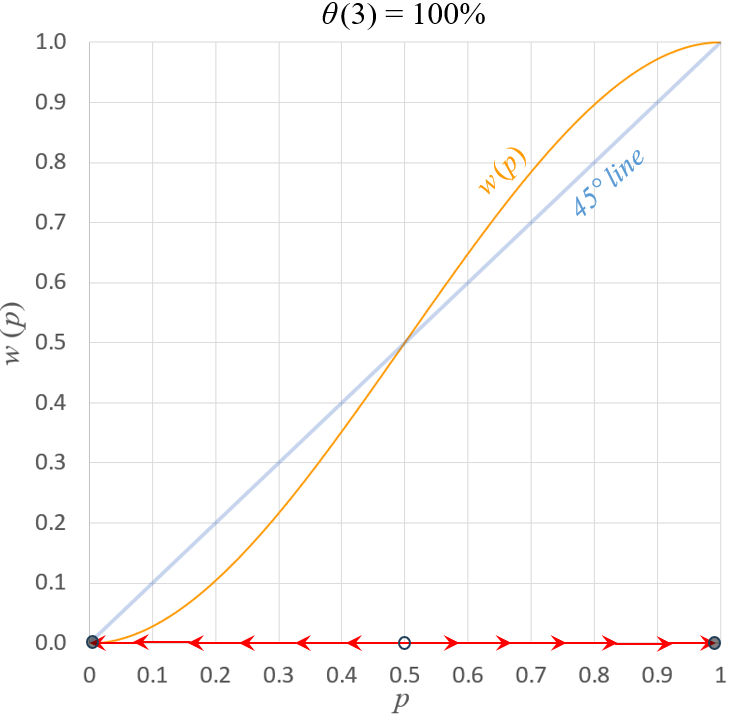} & \includegraphics[scale=0.37]{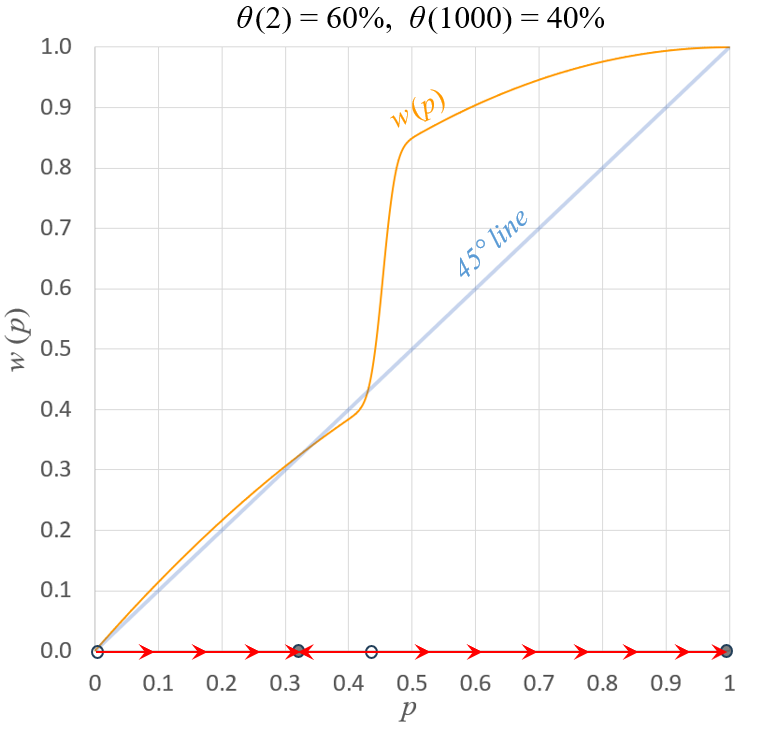}\tabularnewline
\end{tabular}
\par\end{centering}
{\small{}The figure illustrates for each of four environments the one-dimensional phase plot of the sampling dynamics and the $w(p)$ curve relative to the 45$^{\circ}$ line. The intersection points of the $w(p)$ curve and the 45$^{\circ}$ line are the stationary states. A solid (resp., hollow) dot represents an asymptotically stable (resp.,
unstable) stationary state.}{\small\par}
\end{figure}

The top-right panel of Figure \ref{figure-symmetric} illustrates an environment in which all agents sample two actions. In this scenario, only the pure states are stationary, and the population converges from  all initial states to everyone playing  $a$. The intuition is that since $\boldsymbol{a}$ is a risk-dominant equilibrium, a new agent plays $a$  when at least one of the two actions they observe is $a$. This occurs with a probability of $1-(1-p)^2=p(2-p)>p$, which implies that $\dot p>0$ $\forall p>0$. Consequently, the proportion of agents playing $a$ always increases.

The bottom-left panel illustrates the environment in which all agents sample three actions. This environment admits a unique interior stationary state, $p=0.5$, which is unstable; any population beginning to the left (resp., right) of this interior state converges to everyone playing $b$ (resp., $a$). The intuition is that the best response to a sample of three actions is the modal action in the sample. One can show that the probability that $a$ is the modal action is less than $p$ iff $p<0.5,$ which implies that $p=0.5$ is the unique interior stationary state and that it is unstable.

\subsection{Homogeneity and Unstable Miscoordination}
Our first main result shows that the scenarios illustrated in Figure \ref{figure-symmetric}  are the only possible cases for  homogeneous populations. Specifically,  environments in which all agents have the same sample size admit at most one interior stationary state, which must be unstable. This implies that almost all initial states converge
to one of the pure equilibria. 

\begin{thm}
\label{thm:global-stability-pure-sym}\label{thm1}Assume that $\theta\equiv k>1$. There exists at most one interior
stationary state, and this state (if it exists) is unstable.
\end{thm}

\begin{proof}
Define $F_{m}^{k}\left(p\right)\equiv\Pr\left(X\left(k,p\right)\geq m\right)$
as the probability of having at least $m$ successes in $k$ trials
when the probability of success in each trial is $p$. Observe that
$F_{m}^{k}\left(0\right)=0$, $F_{m}^{k}\left(1\right)=1,$ and $\left(F_{m}^{k}\right)'>0$.
It is known that $F_{m}^{k}$ has at most one interior fixed point:

\begin{fact}\label{fact-green}[{\citealp[Theorem 1]{green1983fixed}}]
Fix arbitrary integers $0<m\leq k$. Then there is at most one $p\in\left(0,1\right)$
such that $F_{m}^{k}\left(p\right)=p.$
\end{fact}

The fact that $\theta\equiv k>1$ implies that $w(p)=F_m^k\left(p\right)$
for some $1\leq m\leq k$. 
This implies that any stationary
state ${p}^{\ast}$ must satisfy $F_{m}^{k}\left(p^{\ast}\right)=p^{\ast}$.
Fact  \ref{fact-green} implies that this holds for at most
one interior state $\hat{p}$. Therefore, the stationary state $\hat{p}$ (if it exists) is a neighbor of both pure stationary states $0$ and $1$. 
Observe that if $u<1$ (resp., $u>1$), then an agent plays $a$ (resp., $b$) only when observing at least two $a$'s (resp., $b$'s) in her sample. This implies that if $p=\epsilon$ (resp., $p=1-\epsilon)$, where $\epsilon<<1$, then the probability that the agent plays action $a$ (resp., $b$) is $O(\epsilon^2),$ which implies that the population converges to everyone playing $b$ (resp., $a$) from any sufficiently close state. This implies that state 0 (resp., state 1) is asymptotically stable, which, in turn, implies that the neighboring stationary state $\hat{p}$ is unstable.
\end{proof}

\subsection{Heterogeneity and Stable Miscoordination}
Our second main result shows that introducing heterogeneity in sample size often yields  stable miscoordination. Specifically, Theorem \ref{thm:stable-miscor-sym} shows that for any  $u$ 
and for all distributions in which the sample sizes are at most $u+1$, there exists an interval of $\alpha$-values such that if the sample sizes of a share $\alpha$ of agents are significantly increased (regardless of their initial sample size), the resulting environment admits an asymptotically stable interior state characterized by miscoordination.

\begin{thm}
\label{thm2}\label{thm:stable-miscor-sym}

Fix any $(u,\theta)$ satisfying $1<\max(\emph{supp}(\theta))<\max\left(\frac{1}{u}+1, u+1\right)$. There is  a proportion $\alpha\in(0,1)$ such that substantially increasing the sample sizes of a share $\alpha$ of the agents induces an environment with an asymptotically stable interior state.
\end{thm}
Theorem \ref{thm:stable-miscor-sym} is implied by Theorem \ref{thm:locally-stable_interior_if_pure_is_unstable} and Proposition \ref{prop:stationary-coincide}. In what follows, we sketch the proof (which is illustrated in the bottom-right panel of Figure \ref{figure-symmetric}).
\begin{proof}[Sketch of proof] 
Assume w.l.o.g.  that $u\geq1$. Observe that a new agent with a sample size smaller than $u+1$ plays $a$ iff she observes $a$ in her sample. This implies that $w(\epsilon)>\epsilon$ and $\frac{w(p)}{p}$ is decreasing in $p$. Further, note that an agent with a very large sample almost never plays $a$ as long as $p<p^{NE}.$ 
Thus, there is an interval of $\alpha$'s such that significantly increasing the sample sizes of a share $\alpha$  of the agents induces an environment in which the share
of new agents playing  $a$ is (1) above $\epsilon$
in state $\epsilon$, and (2) below
$\hat{p}$ in some state $\hat{p}$. This implies that there is a stable interior
state between ${\epsilon}$ and ${\hat{p}}$.  
\end{proof}

Theorem \ref{thm:stable-miscor-sym} immediately implies the following corollary:
\begin{cor}
For any $u\neq 1$, there exists a distribution of sample sizes that admits an asymptotically stable interior state with miscoordination.   
\end{cor}

\subsection{Comparison with \cite{oyama2015sampling}}
In the case of a homogenous population with $\theta \equiv k > 1$ (as analyzed in Theorem \ref{thm1}), if $k\leq \left\lceil u\right\rceil$, then $\Pr\left(X\left(k,p\right)\geq\frac{k}{u+1}\right)$
is concave in $p$ and lacks an interior fixed point, implying that $p = 1$ is globally stable (see the top-right panel
of Figure \ref{figure-symmetric}). Conversely, if $k> \left\lceil u\right\rceil$, then $\Pr\left(X\left(k,p\right)\geq\frac{k}{u+1}\right)$ is initially convex, later becoming concave, and has exactly one interior fixed
point which is unstable (see the bottom-left panel of Figure \ref{figure-symmetric}).

In the general case, $w(p)$ is a convex combination of these scenarios, weighted by coefficients
$\theta(k)$. When $\theta(\cdot)$ places large mass on $k> \left\lceil u\right\rceil$, the latter convexity-concavity effect dominates, making the dynamics resemble a standard best-response dynamic with a single interior fixed point which is unstable.
Conversely, when $\theta(\cdot)$ assigns large mass on $k\leq \left\lceil u\right\rceil$, the concavity effect prevails, rendering $p = 1$ globally stable. A sufficient condition for this scenario, given by \citeauthor{oyama2015sampling}'s (\citeyear{oyama2015sampling}) Theorem, can be restated in our setup as follows (assuming $u$ is not an integer):
\[
\sum_{k=1}^{\left\lceil u\right\rceil }\theta\left(k\right)\left(1-\left(\frac{\left\lceil u\right\rceil -1}{\left\lceil u\right\rceil }\right)^{k}\right)>\frac{1}{\left\lceil u\right\rceil}.
\] 

Our analysis, particularly Theorem \ref{thm2}, focuses on cases where this condition is unmet. Specifically, we examine scenarios where the mass on  $k\leq \left\lceil u\right\rceil$ is intermediate,  allowing both the concavity and convexity-concavity effects to influence the dynamics, resulting in multiple interior fixed points (see the bottom-right panel of Figure \ref{figure-symmetric}).

For concreteness, consider the case analyzed in Observation 1 in \citet[p. 257]{oyama2015sampling}, where the population includes two groups: a share $\theta(\left\lceil u\right\rceil)$ of agents with sample size $\left\lceil u\right\rceil$ and the remaining agents with very large samples. \citeauthor{oyama2015sampling} show that global convergence to everyone playing $a$ occurs if $\theta(\left\lceil u\right\rceil)$ exceeds a minimal threshold,
 which equals $\frac{2}{3}$ for $1<u\leq 2.$ The minimal threshold decreases with $u$, and approximately equals 
 $\frac{1.58}{\left\lceil u\right\rceil}$ for large $u$. Applying our proof's method in this scenario shows that for any $u$, there is an open interval starting at $\frac{1}{\left\lceil u\right\rceil},$ such that if $\theta(\left\lceil u\right\rceil)$ is in this interval, then there is an asymptotically stable interior state. For example, when $u=1.5$, there is an asymptotically stable interior state for any $\theta(2)$ in the interval $(0.5,0.6)$.

\section{General Coordination Games}\label{sec:general}
In this section, we relax the assumption of symmetry and analyze general two-player, two-action coordination games. 
We first extend Theorems \ref{thm1} and \ref{thm2} to this setup and then prove two new results for asymmetric coordination games where players disagree on the preferred outcome. Further generalization  to games with more than two actions is presented in Appendix \ref{sec-multiple-actions}, and to those with more than two players in Appendix \ref{sec-multiple-players}. 
\subsection{Model}
We adapt our definitions to allow asymmetry between the two players. Let $i\in\left\{ 1,2\right\} $ denote
one of the players (``she''), and $j$ her opponent (``he''). 
For each $i$, let $A_{i}=\left\{ a_{i},b_{i}\right\}$ represent player $i$'s actions.  The standard two-parameter payoff matrix for a coordination
game is in Table \ref{tab:standard-coordination-game-gen}: players receive 0 if they miscoordinate, 1 if they coordinate on both playing $\textbf{b}\equiv\left(b_{1},b_{2}\right)$, and 
 $u_i>0$ if they coordinate on both playing $\textbf{a}\equiv\left(a_{1},a_{2}\right)$. 
By relabeling the actions, we can assume w.l.o.g. that action profile $\textbf{a}$ is Player 1's weakly preferred outcome, i.e., $u_1\geq 1.$
We define the coordination game as (1)  \emph{symmetric} if $u_{1}=u_{2}$, 
and (2) \emph{antisymmetric} if $u_{1}=\frac{1}{u_{2}}.$ Antisymmetric coordination games can be interpreted as \emph{battle of the sexes} games, where
Player 1's preference for \textbf{a} over \textbf{b} is as strong as Player 2's preference for $\textbf{b}$ over $\textbf{a}$.

\begin{table}
\caption{\label{tab:standard-coordination-game-gen}Standard Representation
of a Two-Action Coordination Game
(\textcolor{blue}{$u_1\geq1$, }\textcolor{red}{$u_2$}~>~0)}

\centering{}%
\begin{tabular}{|c|c|c|}
\hline 
 & \textcolor{red}{\emph{$a_{2}$}} & \textcolor{red}{\emph{$b_{2}$}}\tabularnewline
\hline 
\textcolor{blue}{\emph{$~~a_{1}~~$}} & \emph{\Large{}$\,\,{}_{\underset{\,}{{\color{blue}u_{1}}}}{\color{red}\,^{\overset{\,}{u_{2}}}}$~~} & \emph{\Large{}${\color{blue}_{\underset{\,}{0}}}\,\,^{{\color{red}\overset{\,}{0}}}$}\tabularnewline
\hline 
\textcolor{blue}{\emph{$~~b_{1}~~$}} & \emph{\Large{}$_{\underset{\,}{{\color{blue}0}}}\,\,^{\overset{\,}{{\color{red}0}}}$} & \emph{\Large{}$_{{\color{blue}1}}\,\,\,\,^{{\color{red}1}}$}\tabularnewline
\hline 
\end{tabular}
\end{table}

We identify each mixed action by the probability assigned to the first action ($a_{i}$), and denote it by $p_{i}\in\left[0,1\right]$. A state $\textbf{p}$ is \emph{interior} if $p_1,p_2\in(0,1)$. The coordination game admits three Nash equilibria: 
two pure 
equilibria, $\textbf{a}$ and $\textbf{b},$ and an interior (mixed) equilibrium  $\boldsymbol{p}^{NE}=\left(\frac{1}{1+u_{2}},\frac{1}{1+u_{1}}\right)$ with miscoordination. Appendix \ref{subsec-general-coord} shows that this two-parameter representation captures all two-action games with two strict equilibria 
(and, in particular, hawk--dove games).

 \subsection{Evolutionary Dynamics}
In general games, agents are assigned one of two roles: Player 1 or Player 2. We model this interaction using a two-population dynamic, with two distinct unit-mass populations where agents from Population 1 are randomly matched with agents from Population 2.  Aggregate behavior at time $t\in\mathbb{R}^{+}$ is described by
a \emph{state} $\mathbf{p}\left(t\right)\equiv\left(p_{1}\left(t\right),p_{2}\left(t\right)\right)\in\left[0,1\right]^{2}$,
where $p_{i}\left(t\right)$
is the share of agents playing action $a_{i}$ in population $i$ at time $t$. A state $\mathbf{p}$
is \emph{interior} (i.e., mixed) if $p_{1},p_{2}\in(0,1).$ 

As in the baseline model, agents die at a constant rate of 1, and are replaced by new agents. Let $\theta_{i}\in\Delta\left(\mathbb{N}_{>0}\right)$ denote the distribution
of sample sizes of new agents in population $i$, with $\theta_{i}$
having finite support. We redefine an \emph{environment} to be a pair
$\left(\boldsymbol{u},\boldsymbol{\theta}\right)=\left(\left(u_{1},u_{2}\right),\left(\theta_{1},\theta_{2}\right)\right)$,
where $\left(u_{1},u_{2}\right)$ are the payoffs, and $\left(\theta_{1},\theta_{2}\right)$ are
the distributions of sample sizes. An environment is symmetric if $\theta_1=\theta_2$ and $u_1=u_2$.

Each new agent of population $i$ with sample size $k$ samples $k$ actions from the other population and plays the best-response action to the  sample.  Using similar calculations as in Section \ref{subsec-one-population-dynamic}, we derive the sampling dynamics for the environment $\left(\boldsymbol{u},\boldsymbol{\theta}\right)$, where  the random variable $X(k,p_{j})\sim Bin\left(k,p_{j}\right)$ denotes the number of $a_{j}$'s in the sample:
\begin{equation}
w_{i}\left(p_{j}\right)=\sum_{k\in\text{supp}(\theta_{i})}\theta_{i}(k)\cdot\Pr\left(X\left(k,p_{j}\right)\geq\frac{k}{u_{i}+1}\right),\,\,\,\,\,\,\,\,\,\,\,\, \dot{p}_{i}=w_{i}(p_j)-p_{i}.\label{eq:action-sampling-dyanmics-g}
\end{equation}

Observe that  $w_{i}\left(p_{j}\right)$ is a strictly increasing polynomial of finite degree
$\max(\textrm{supp}(\theta_{i}))$, and that $w_i(0)=0$ and $w_i(1)=1$. This implies
that the inverse function $w_{i}^{-1}:[0,1]\rightarrow[0,1]$ exists, that it
is continuously differentiable, and that $w_{i}^{-1}(0)=0$ and $w_{i}^{-1}(1)=1.$ 

\subsection{Preliminary Analysis and Examples}
 Eq. (\ref{eq:action-sampling-dyanmics-g}) immediately implies that $\dot p_i>0$ iff $w_i(p_j)>p_i$, $\dot p_i<0$ iff $w_i(p_j)<p_i$, and that $\boldsymbol{p}$ is stationary iff $p_i=w_i(p_j)$ for each player $i$ (which implies that $p_i=w_i(w_j(p_i))$ for each player $i$).  This, in turn, implies:
 \begin{prop}[Adaptation of Fact \ref{fact-w-stable}]\label{prop-w-stable}~\\ 
 \begin{enumerate}
  \vspace{-28px}
    \item State \textbf{0} is asymptotically stable iff $w_2(p_1)<w_1^{-1}(p_1)$ in a right neighborhood of 0.
    \vspace{-3px}
    \item State \textbf{1} is asymptotically stable iff $w_2(p_1)>w_1^{-1}(p_1)$ in a left neighborhood of 1.
    \vspace{-3px}
   \item An interior stationary state $\textbf{p}$ is asymptotically stable iff $w_2(p_1)>w_1^{-1}(p_1)$ in a left neighborhood of $p_1$ and  $w_2(p_1)<w_1^{-1}(p_1)$  in a right neighborhood of $p_1$. 
   \vspace{-3px}
    \item A neighboring stationary state of an asymptotically stable state is unstable.\footnote{Stationary states $\underline{\boldsymbol{p}}, \boldsymbol{\overline{p}}$ are neighbors if no stationary state $\hat{\boldsymbol{p}}$ exists with $\underline{p}_{i}<\hat{p}_{i}<\overline{p}_{i}$.} 
    \vspace{-3px}
    \item $\lim_{t\rightarrow\infty}\mathbf{p}\left(t\right)$
    exists and is a stationary state for any $\textbf{p}(0).$ 
    \vspace{10px}

\end{enumerate}
 \end{prop}
Proposition \ref{prop-w-stable} is proven in Appendix \ref{app-proof-of-Prop-1} and illustrated in Figure \ref{figure-asymmetric}. The figure presents the phase plots of the dynamics along with the $w_i(p_j)$ curves for four environments, all having $u_1 = u_2 = 1.2$ but differing in their distributions of sample sizes (the environments are the two-population counterparts of those in Figure \ref{figure-symmetric}). The orange dotted curve represents $w_2(p_1)$, indicating the states where $\dot{p}_{2}=0$. In states above
this curve, $\dot{p}_{2}<0$, and in states below it, $\dot{p}_{2}>0$).
The purple solid curve represents either $w_1(p_2)$ (plotted against $p_2$ in the \textit{Y}-axis) or  $w_{1}^{-1}\left(p_{1}\right)$ (plotted against $p_1$ in the \textit{X}-axis). In states to the right of this curve, $\dot{p}_{1}<0$, and in states to the left, $\dot{p}_{1}>0$.

\begin{figure}[h]
\caption{Phase Plots and $w_i\left(p_j\right)$ Curves for Various Environments with
$u_1=u_2=1.2$}\label{figure-asymmetric}

\begin{centering}
\includegraphics[scale=0.5]{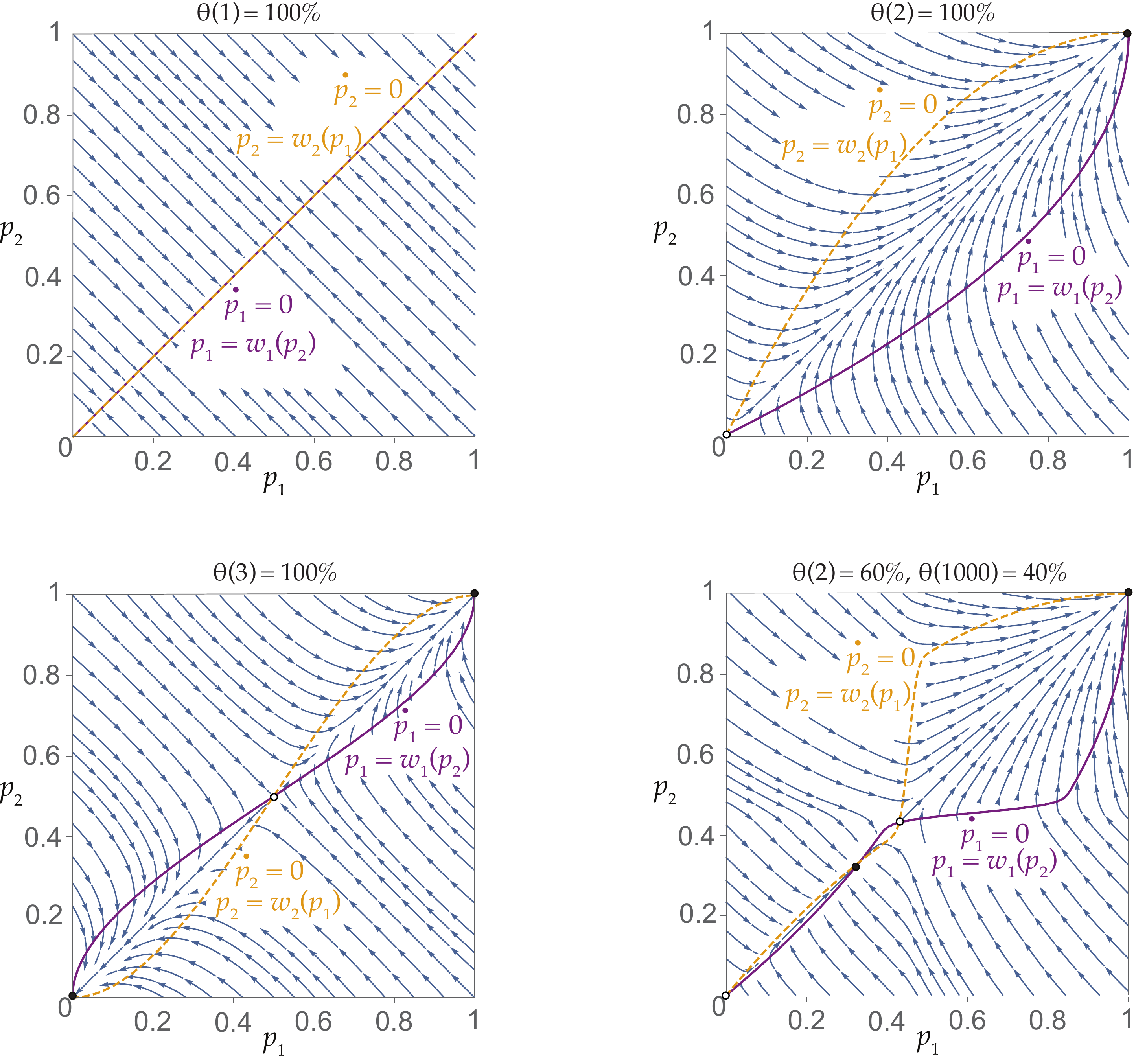}
\par\end{centering}
{\small{}The figure illustrates for each of four environments the two-dimensional phase plot of the sampling dynamics and the $w_i(p_j)$ curves. The intersection points of the two $w_i(p_j)$ curves are the stationary states. A solid (resp., hollow) dot represents an asymptotically stable (resp.,
unstable) stationary state.}{\small\par}
\end{figure}
The top-left panel shows that with all agents having a sample size of 1, all symmetric states $(p,p)$ are stationary, as in the baseline model. By contrast, when some agents in one of the populations sample multiple actions, $w_1(w_2(p))$ becomes a polynomial of positive degree greater than one. This implies that $w_1\circ w_2$ has a finite number of fixed points, which means the dynamics have a finite number of stationary states:

\theoremstyle{plain}
\newtheorem*{fact2prime}{Fact 2'}
\begin{fact2prime}\label{fact-fintie-asym} 
If  $\theta_i\equiv1$ for each player $i$, a state is stationary iff it is symmetric. Otherwise,  the number of stationary states is finite.
\end{fact2prime}

Symmetric coordination games should be modeled as single-population dynamics  if agents cannot condition their play on their role (Player 1 or Player 2), and as two-population dynamics if they can.  Notably, in all scenarios illustrated in Figures \ref{figure-symmetric} and \ref{figure-asymmetric}, both types of dynamics produce the same set of asymptotically stable states, and, hence, the same predictions about the long-run behavior in symmetric coordination games. Our next result shows that this holds true in general (see Appendix \ref{proof-prop-2-sym} for the proof).\footnote{The equivalence between one-population and two-population dynamics also holds in the original payoff matrix with four parameters (as shown on the right side of Table \ref{tab-symmetric}), provided that the two strict equilibria are on the main diagonal. If the strict equilibria are off the diagonal, they become infeasible in a single-population context. Consequently, in one-population dynamics, hawk--dove games transform into anti-coordination games, leading to convergence to an interior state characterized by miscoordination, whether under monotone or sampling dynamics. By contrast, hawk--dove games remain coordination games in two-population dynamics, as detailed in Appendix \ref{subsec-general-coord}.} The intuition is that in any asymmetric state (say with $p_1<p_2$) the share of new agents playing $a_i$ is larger in population 1, which pushes the populations toward the 45$^\circ$ line of symmetric state. Once in a symmetric state, the dynamics (Eq. \ref{eq:action-sampling-dyanmics-g}) coincide with the one-population dynamics, Eqs. \eqref{eq:dynamics-sym}--\eqref{eq:action-sampling-dyanmics}.

\begin{prop}\label{prop:stationary-coincide}
    Let $(\textbf{u},\boldsymbol{\theta})$ represent a symmetric environment. A state $\textbf{p}$ is  asymptotically stable under two-population dynamics iff $\textbf{p}$  is symmetric (i.e., $p\equiv p_1=p_2$) and $p$ is  asymptotically stable under one-population dynamics $(u,\theta)$.
\end{prop}

\subsection{Generalizing Theorems \ref{thm1} and \ref{thm2}}
Theorem \ref{thm:global-stability-pure} shows that in environments where all agents within each population have the same sample size, there is at most one interior stationary state, which is unstable. Consequently, almost all initial states converge to one of the pure equilibria.

\vspace{5px}
\begin{primedtheorem}\label{thm:global-stability-pure}
\emph{Assume that $\theta_{i}\equiv k_{i}>1$
for each $i\in\left\{ 1,2\right\} $. There exists at most one interior stationary state, and this state (if it exists) is unstable.}
\end{primedtheorem}

The proof, detailed in Appendix \ref{app-proof-of-thm1'}, extends Fact \ref{fact-green} by showing that a composition of two cumulative binomial distributions has at most one interior fixed point. This result, presented as Proposition \ref{prop-binomial} in Appendix \ref{subsec:General-Result-forbinomial}, may be of independent interest.

Next, we show that in many environments, there is an interval of $\alpha_i$ values where increasing the sample sizes of $\alpha_i$ agents in population $i$ (regardless of their current sample sizes) leads to an environment with an asymptotically stable interior state characterized by miscoordination. This is illustrated in the bottom-right panel of Figure \ref{figure-asymmetric}.

\vspace{5px}
\begin{primedtheorem}\label{thm:locally-stable_interior_if_pure_is_unstable}
\emph{In an environment where $1<\max(\emph{supp}(\theta_i))< u_i+1$ for each population $i$, there is a proportion $\alpha_i\in (0,1)$ such that substantially increasing the sample sizes of a share $\alpha_i$ of the agents in each population $i$ induces an asymptotically stable interior state.}
\end{primedtheorem}

See Appendix \ref{Proof-Thm-2'} for the proof. The following subsections present new results for asymmetric coordination games with different preferred outcomes for the players.

\subsection{Stable States with a High Level of Miscoordination\label{subsec:Heterogeneity-no-dominant}}
Theorem \ref{thm:locally-stable_interior_if_pure_is_unstable} does not provide a uniform minimal bound for the level of miscoordination in the stable interior state. In the next result, we provide sufficient conditions that ensure high levels of miscoordination (>50\%) in the stable interior state in games where the two populations have different preferred coordinated outcomes (i.e., $u_2 < 1 < u_1$).

\begin{figure}[h]
\begin{centering}
\caption{Illustrative Phase Plots for Theorems \ref{thm:locally-stableinterior_ASD} and \ref{thm:global-mixed}
 \label{fig:fig3-local-stablity}}
\includegraphics[scale=0.59]{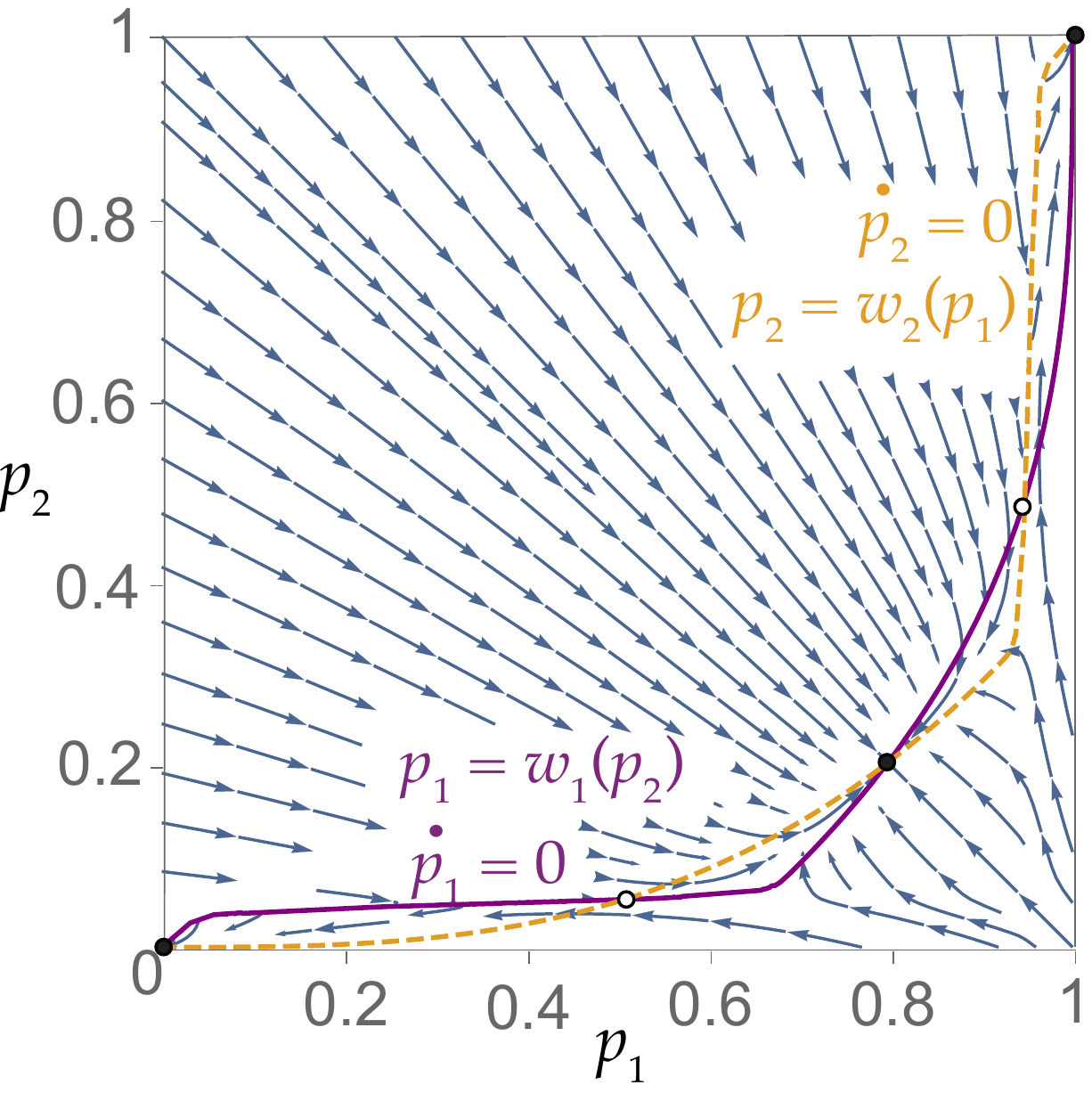}~~~~~~~~~\includegraphics[scale=0.59]{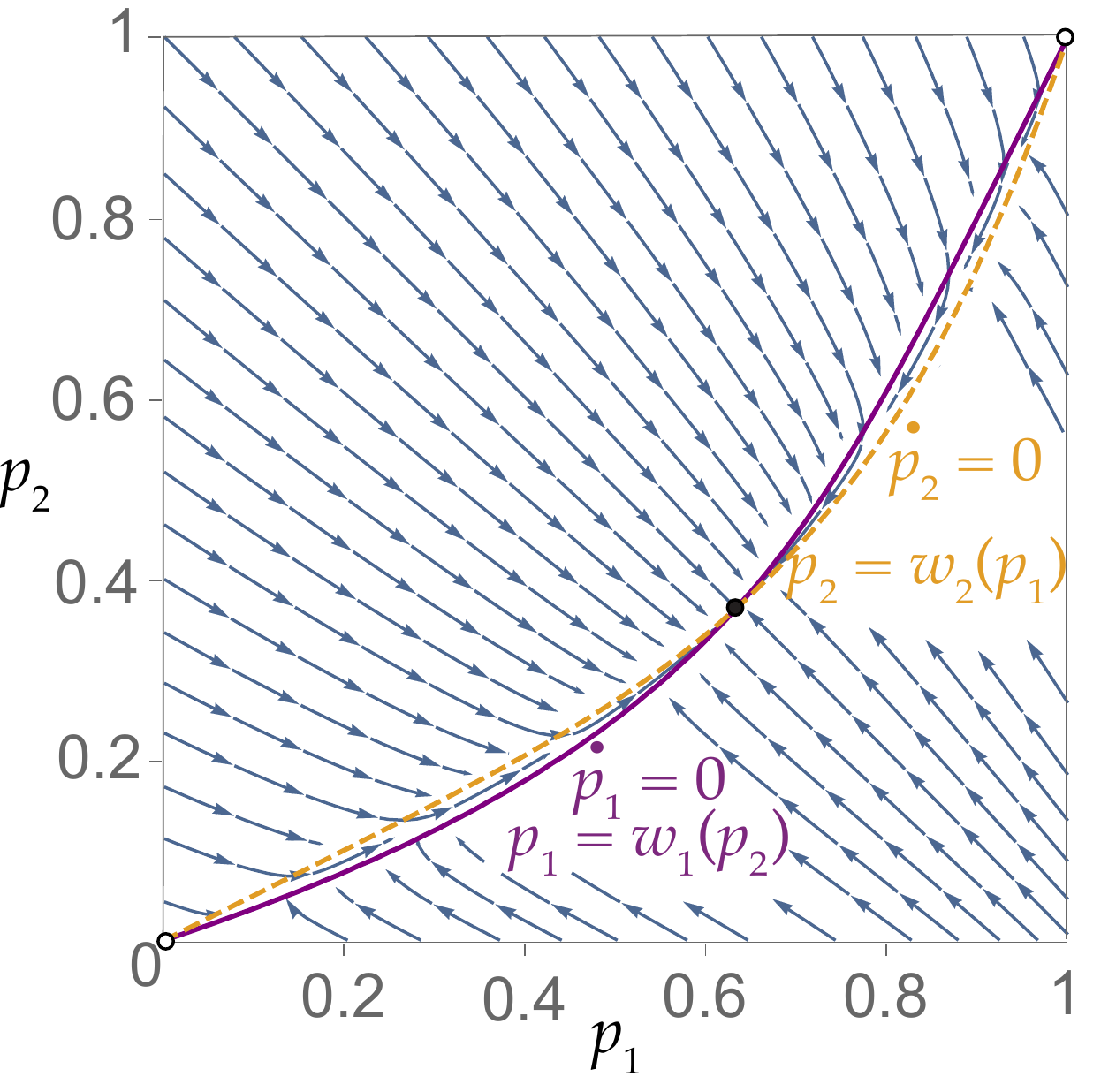}
\par\end{centering}
\small{The left panel illustrates Theorem \ref{thm:locally-stableinterior_ASD}. It shows an 
antisymmetric game with
$u_{1}=\frac{1}{u_{2}}=20$, where 50\% (resp.,
50\%) of the agents have sample size 3 (resp., 1,000). The mixed Nash
equilibrium of this game is $\left(0.95,0.05\right)$.
The environment admits three interior stationary states: two unstable states at $\left(0.47,0.05\right)$ and  $\left(0.95,0.53\right)$), and a stable
 state at $\left(0.77,0.23\right)$ 
with a miscoordination probability of $65\%=0.77^2+0.23^2$.} The right panel illustrates  Theorem \ref{thm:global-mixed}. It shows an 
antisymmetric game with
$u_{1}=\frac{1}{u_{2}}=5$, where 50\%  of the agents have sample size 1 and 50\% have sample size 5. The environment admits a globally stable interior state $(0.63,0.37)$.
\end{figure}
   
The \emph{miscoordination probability} induced by state $\mathbf{p}$ is the fraction of matched agent pairs (one from each population) who miscoordinate by selecting different actions. This miscoordination probability is given by $p_1\cdot (1-p_2)+p_2\cdot (1-p_1).$

\begin{thm}
\label{thm:locally-stableinterior_ASD}\label{thm3} For any sample size distribution profile, if  $u_1$ and $\frac{1}{u_2}$ are sufficiently large, then significantly increasing the sample size of  half of the agents in each population induces an asymptotically stable interior state with a miscoordination probability of at least 50$\%.$ 
\end{thm}

\begin{proof}[Sketch of proof]
The argument is illustrated in the left panel of Figure \ref{fig:fig3-local-stablity}. 
Let $\theta'_i$ (resp., $\theta_i$) be the sample size distribution in population $i$ after (resp., before) the sample sizes of half of the agents are substantially  increased. Observe that
(1) $w_2^{\theta'_2}(\frac{1}{2})\approx\frac{1}{2}w_2^{\theta_2}(\frac{1}{2})$ and (2) $\left(w_1^{\theta'_1}\right)^{-1}(\frac{1}{2})\approx p_2^{NE}<<1~\Rightarrow~ w_2^{\theta'_2}\left(\frac{1}{2}\right)>\left(w_1^{\theta'_1}\right)^{-1}\left(\frac{1}{2}\right).$ Analogous arguments imply that $w_1^{\theta'_2}(\frac{1}{2})<(w_2^{\theta'_2})^{-1}(\frac{1}{2}).$  Thus, the curve $w_2^{\theta'_2}$ is above  (below) $\left(w_1^{\theta'_2}\right)^{-1}$ at $\frac{1}{2}$ ($\left(w_1^{\theta'_1}\right)^{-1}\left(\frac{1}{2}\right)>\frac{1}{2}$), which implies that there is a stationary state $\hat{\boldsymbol{p}}$ between $\frac{1}{2}$ and $\left(w_1^{\theta'_1}\right)^{-1}\left(\frac{1}{2}\right)$ such that the curve $w_2^{\theta'_2}$ is above (resp., below) $(w_1^{\theta'_2})^{-1}$ in a left (resp., right) neighborhood of $\hat{\boldsymbol{p}}$. Proposition  \ref{prop-w-stable}, then, implies that  $\hat{\boldsymbol{p}}$ is asymptotically stable.
Note that $\hat{\boldsymbol{p}}$ is in the bottom-right quarter of the unit square, which implies that its induced miscoordination probability is greater than 50\%. See Appendix \ref{subsec:Proof-of-Theorem-3} for a formal
proof.\qedhere
\end{proof}

\subsection{Global Convergence to Miscoordination\label{sec:global-convergence-ro-pure-states}}
Theorems \ref{thm:locally-stable_interior_if_pure_is_unstable} and \ref{thm3} demonstrate that many coordination games and heterogeneous distributions result in locally stable interior states, which often coexist with stable pure states.\footnote{In these cases, the long-run behavior depends on initial conditions, and the basin of attraction of the asymptotically stable interior state is often large (e.g., 87\% in the right panel of Figure \ref{fig:fig3-local-stablity}).} In this subsection, we characterize when an interior state with miscoordination is \emph{globally} stable, i.e., when it attracts almost all interior states. We show that this occurs for a small but significant set of environments where the $u_i$ values are on opposite sides of, and sufficiently far from, 1, and where a sufficiently large share of agents (but not all) in each population have a sample size of 1. By contrast, our previous results did not require the presence of agents with a sample size of 1.

In order to present our result, it would be helpful to define the $m$-truncated expectation as the expected value
when ignoring values larger than $m.$ Formally:
\begin{defn}
The $m$\emph{-truncated expectation} $\mathbb{E}_{\leq m}$ (resp.,
$\mathbb{E}_{<m}$) of distribution $\theta_{i}$ with support on
the natural numbers is defined as follows:\footnote{Observe that in our notation the parameter $k$ takes only positive
integer values (although we allow the upper bound $m$ to be a noninteger).}\\ 
\vspace{-20px}
$$\mathbb{E}_{\leq m}\left(\theta_{i}\right)=\sum_{1\leq k\leq m}\theta_{i}\left(k\right)\cdot k
\,\,\,\,\,\,\,\,\,\,\,\,\,\,\textrm{ (resp., } \mathbb{E}_{<m}\left(\theta_{i}\right)=\sum_{1\leq k<m}\theta_{i}\left(k\right)\cdot k).$$ 
\end{defn}

Our final main result shows that populations
converge to an interior stationary
state from almost any initial state if, and essentially only if, the product
of the truncated expectation of the distribution of sample sizes in each population and
the  share of agents with sample size 1 in the other population exceeds 1. 
\begin{thm}
\label{thm:global-mixed} \label{thm4}

\begin{enumerate}
\item \textbf{Global convergence to miscoordination}: Assume that \vspace{-10px}
\[
\theta_{1}\left(1\right)\cdot\mathbb{E}_{<\frac{1}{u_{2}}+1}\left(\theta_{2}\right)>1\,\,\textrm{\,\,and}\,\,\,\,\theta_{2}\left(1\right)\cdot\mathbb{E}_{
< u_{1}+1}\left(\theta_{1}\right)>1.
\]
If $\boldsymbol{p}(0)\notin\left\{ (0,0),(1,1)\right\} $, then $\lim_{t\rightarrow\infty}\mathbf{p}\left(t\right)\notin\left\{ (0,0),(1,1)\right\} $. 
\item \textbf{Local convergence to coordination}: Assume that 
\[
\theta_{1}\left(1\right)\cdot\mathbb{E}_{
\leq\frac{1}{u_{2}}+1}\left(\theta_{2}\right)<1\,\,\,\,\textrm{or}\,\,\,\,\theta_{2}\left(1\right)\cdot\mathbb{E}_{\leq u_{1}+1}\left(\theta_{1}\right)<1.
\]
Then at least one of the pure equilibria is asymptotically stable.
\end{enumerate}
\end{thm}
\begin{proof}[Sketch of proof. See Appendix \ref{subsec:Proof-of-Theorem-4} for a formal proof.]
Consider a slightly perturbed state $(1-\epsilon_{1},1-\epsilon_{2})$
near $\boldsymbol{a}=\left(1,1\right)$, where almost all agents play action $a_{i}$. The
probability of two rare actions ($b_{i}$) appearing in a 
new agent's sample is negligible, on the order of $O(\epsilon_{i}^{2})$.
For a new agent with a sample size of $k$, the probability of
observing a rare action is approximately $k\cdot\epsilon_{j}.$
This rare occurrence changes the new agent's perceived best response in population $i$ iff $1>(k-1)\cdot u_i$, which simplifies to $k<\frac{1}{u_i}+1$.

Thus, the
probability that a new agent in population $i$ adopts
the rare action $b_{i}$ is $\mathbb{E}_{<\frac{1}{u_{i}}+1}\left(\theta_{i}\right)\cdot\epsilon_j$.
The product of the shares of new agents adopting
the rare action in both populations is then $\mathbb{E}_{<\frac{1}{u_{1}}+1}\left(\theta_{1}\right)\cdot\epsilon_{2}\cdot\mathbb{E}_{<\frac{1}{u_{2}}+1}\left(\theta_{2}\right)\cdot\epsilon_{1}$.
This product determines whether the share of agents playing rare actions increases (indicating instability) or decreases (indicating asymptotic stability). Specifically, the system is unstable if $\mathbb{E}_{<\frac{1}{u_{1}}+1}\left(\theta_{1}\right)\cdot\mathbb{E}_{<\frac{1}{u_{2}}+1}\left(\theta_{2}\right)>1$, and asymptotically stable if the product is less than 1.

Observe that since we assume that $u_{1}\geq1$, it follows that $\frac{1}{u_{1}}+1\leq2$, implying that $\theta_{1}\left(1\right)=\mathbb{E}_{<\frac{1}{u_{1}}+1}\left(\theta_{1}\right)$. Thus, state $\textbf{a}$ becomes unstable if $\theta_1(1)\cdot\mathbb{E}_{<\frac{1}{u_{2}}+1}\left(\theta_{2}\right)>1$. By an analogous argument, state $\textbf{b}$ is unstable if $\theta_2(1)\cdot\mathbb{E}_{<u_1+1}\left(\theta_{2}\right)>1$. If both pure states are unstable, then we get global convergence to interior states with miscoordination.
\qedhere
\end{proof}
Observe that if both players prefer the same coordinated outcome (i.e., if $u_1,u_2>1$, then $\theta_{1}\left(1\right)\cdot\mathbb{E}_{
\leq\frac{1}{u_{2}}+1}\left(\theta_{2}\right)=\theta_{1}\left(1\right)\cdot\theta_{2}\left(1\right)<=1,$ with a strict inequality if either population has agents who sample more than one action). In this case, Part 2 of Theorem \ref{thm4} implies that the preferred state $\textbf{a}$ must be asymptotically stable. Therefore, global convergence to miscoordination can only occur if the players prefer different coordinated outcomes (i.e., if $u_2<1<u_2)$.

Theorem \ref{thm:global-mixed}  shows that global convergence to miscoordination requires heterogeneity in sample size within each population, including both agents with a sample size of one and those with larger samples (but not too large, as they have to allow a single observation of a rare action to influence
behavior). Specifically, in each population, the product of (1) the share of agents with a sample size of one and (2) the truncated expected sample size must be sufficiently large. Note that the farther the $u_i$'s are from one, the less restrictive the required truncated expected value becomes.

\begin{rem}

Several papers have  discussed the stability conditions of equilibria in terms of their level of risk-dominance (referred to as $q$-dominance  in \citealp{morris1995p,oyama2015sampling}). To facilitate comparison with the literature and provide a clearer interpretation of our results in the original (8-parameter) representation of a coordination game, we rephrase our results in terms of $q$-dominance in Appendix \ref{subsec:stability-pure-q-dominance}.
\end{rem}

\section{Discussion}\label{Sec-discussion}
\subsection{Extensions \label{Sec:Extensions}}

Our baseline model focuses on two-player two-action coordination games and on a specific class of sampling dynamics. Online Appendix \ref{online-appendix} extends our analysis 
to more than two players, more than two actions, and other classes of learning dynamics.

Appendix \ref{subsec-general-coord} shows that our baseline model captures all two-player two-action coordination games. Appendix \ref{sec-multiple-actions} shows that our main results  remain qualitatively similar in coordination games with more than two actions under the assumption that the payoffs on the main diagonal are positive, while all off-diagonal payoffs are set to zero. This important class of multi-player coordination games is called contracting games; see, e.g., \cite{young1998conventional}.  Appendix \ref{sec-multiple-players} generalizes our analysis to  the commonly studied minimum-effort coordination games with an arbitrary number of players  (\citealp{vanHuyck1990tacit}). 
In Appendix \ref{sec:logit} , we show that our result holds for logit dynamics (\citealp{fudenberg1995consistency}), where agents play a noisy best response. With homogeneous noise levels, stable miscoordination requires implausibly high noise. However, with heterogeneous noise levels, stable miscoordination can occur with moderate noise.

\subsection{Empirical Relevance}\label{subsec-empirical}
To demonstrate the empirical relevance of our findings, we need real-life scenarios with: (1) heterogeneous sample sizes, (2) persistent miscoordination within an interior state, and 
(3) a significant number of agents not maximizing payoffs. We suggest that these conditions may be present in bargaining situations in housing and used-car markets.

Coordination games model simple bargaining in these markets, where agents choose between high and low pricing strategies. For instance, a seller may demand a high price or settle for a low price, while a buyer may agree to a high price or insist on a low price, similar to hawk--dove coordination games (see Example \ref{exa:motivating-hawk--dove} in Appendix \ref{subsec-general-coord}). Markets often exhibit sample size heterogeneity, with professional players (real-estate investors and used-car dealers)
 having large samples and inexperienced players (those who have only bought or sold houses
or cars a couple of times) relying on anecdotal evidence.

\citet{koster2024housing} found that many Dutch housing sellers set low list prices, which aligns with profit maximization only if they have high annual discount factors (up to 50\%). This behavior matches a stable interior stable state in our model where some sellers use small, noisy samples to set non-payoff-maximizing low prices.\footnote{\citeauthor{koster2024housing}'s data does not allow separating between different types of sellers. By contrast, \citet{genesove2001loss} do allow separating the sellers into owner-occupants and investors, and show that the two groups have systematic differences in their list prices.
} 

\citet*{larsen2021efficiency} found that bargaining fails in about 35\% of wholesale used-car auctions where the highest bid is below the reserve price, consistent with a stable interior state with miscoordination. This inefficiency results in significant losses, representing 15\%–20\% of ex-post trade gains. While incomplete information about opponents' preferences is a common explanation for miscoordination, \citeauthor{larsen2021efficiency} suggests that it accounts for only a small portion of the observed inefficiencies.

\subsection{Testable Experimental Prediction}\label{subsec-experimental}

Our model offers a novel testable prediction: miscoordination can persist when agents have varying levels of information about opponents' aggregate behavior.  \cite{brunner2011stationary} compared learning models' predictive power using data from \cite{mckelvey2000effects}, \cite{goeree2003risk}, and \cite{selten2008stationary}, where agents (1) repeatedly play two-player games with anonymous opponents, and (2) rely on feedback from their most recent opponent. They show that sampling dynamics predict behavior as well as quantal response equilibrium in all games, outperform Nash equilibrium, and improve further in less noisy second halves of the experiments (\citealp{brunner2011stationary}, Figure 4).

\citeauthor{brunner2011stationary}'s results suggest that sample sizes of 3–12 fit the data best, aligning with typical estimates of short-term memory capacity. Global convergence to miscoordination (Theorem \ref{thm:global-mixed}) requires some agents to have a sample size of 1, which could be induced by varying the underlying game each round. In each round, subjects are informed of the current game’s payoff matrix and reminded of the opponent's behavior from the last time the same game was played. Many would likely rely only on this feedback (sample size 1), while a few might remember earlier feedback.

 \citeauthor{lyu-in-Xu-2022} (\citeyear{lyu-in-Xu-2022}) found that sampling dynamics predicted  80\% of subjects' behavior. Deviations, especially early on, were mainly due to subjects choosing the Pareto-dominant action to teach others to move from the Pareto-dominated equilibrium. These teaching incentives, which often succeeded, would likely have been reduced and the fit of sampling dynamics improved if (1) the equilibria had not been Pareto-ranked (e.g., battle of the sexes) or (2) the matching groups had been larger.
Testing the predictions of Theorems \ref{thm:locally-stable_interior_if_pure_is_unstable} and \ref{thm:locally-stableinterior_ASD}
regarding the local stability of miscoordination requires both an appropriate initial state (as in \citealp{lyu-in-Xu-2022}) and a setup where some agents receive feedback about the entire population's behavior, while others receive feedback only about their matched opponents.

\subsection{Asymptotic Behavior of Sampling Dynamics}
It is well known that sampling dynamics converge to best-response dynamics for large sample sizes, under fixed game parameters. Specifically, there exists a sufficiently large minimal sample size $\bar k$ such that, if all agents have sample sizes greater than $\bar k$, the dynamics resemble perturbed best-response dynamics, with globally stable pure equilibria.
By contrast, Theorem \ref{thm:locally-stableinterior_ASD} reveals that \emph{sampling dynamics remain qualitatively different from best responding
when considering a different order of limits}. Specifically, if one fixes an arbitrarily large minimal sample size $\bar k$, then there exists a bound on the game parameters $\bar u$, such that if $\frac{1}{u_2},u_1>\bar u$, then  there exists an asymptotically stable interior state with a high level (>50\%) of miscoordination for populations in which all agents have sample sizes larger than $\bar k$.
\subsection{Best Experienced Payoff-Sampling Dynamics}\label{subsec-BEP}
A different type of sampling dynamics in the literature is best experienced payoff-sampling dynamics (\citealp* {osborne1998games,mantilla2018efficiency,sandholm2019best,sandholm2020stability,Raj,arigapudi2020instability, izquierdo2022stability}). Under these dynamics, each new agent  with sample size $k$
observes the payoffs obtained by agents of her own population $i$ in (1) $k$ interactions
in which the incumbents played $a_i$, and (2) $k$ interactions in which the incumbents
played $b_i$. Following these observations, the new agent adopts the action that yielded the higher mean payoff. Arguably, these dynamics capture learning in situations in which new agents either may not be able to observe opponents’ actions or may lack information about the payoff matrix.\footnote{
\citealp{cardenas2015stable} show that best experienced payoff-sampling dynamics provide a good fit for explaining experimental behavior in common pool resource games.}

\cite{arigapudi2022sampling} study the stability of miscoordination under best experienced payoff-sampling dynamics in the subclass of asymmetric coordination games known as hawk--dove games.\footnote{
The dependency of best experienced payoff-sampling dynamics on payoffs per se, rather than on payoff differences, makes it more challenging to extend the results to other types of coordination games.} The main finding is that heterogeneity facilitates stable miscoordination also in best experienced payoff-sampling dynamics, but unlike in our setup, heterogeneity is not necessary for asymptotically stable miscoordination.

Specifically, comparing \citeauthor{arigapudi2022sampling}'s Theorem 1 and our Theorem \ref{thm4} shows that in any heterogeneous environment, if miscoordination is globally stable  under the sampling dynamics studied in this paper, then miscoordination is also globally stable under the best experienced payoff-sampling dynamics. In addition, under the best experienced payoff-sampling dynamics, global stability is achieved in heterogeneous populations if a sufficiently large share of agents have a sample size of either 1 or 2, whereas under the sampling dynamics studied in this paper, global stability is achieved only if a sufficiently large share of agents has a sample size of 1.

\citeauthor{arigapudi2022sampling}'s Theorem 2 demonstrates that, under best experienced payoff-sampling dynamics, homogeneous populations with relatively small sample sizes can induce asymptotically stable miscoordination in some environments. This contrasts with our setup, where Theorem \ref{thm:global-stability-pure} shows that heterogeneity is necessary for achieving asymptotically stable miscoordination.

The intuition is as follows. Under the sampling dynamics studied in this paper, the learning dynamics $w_i(p_j),$ Eq. \eqref{eq:action-sampling-dyanmics-g} depend on a single binomial random variable, which induces a unimodal sensitivity to perturbations, keeping the interior state close to the peak of the slope. This implies that small perturbations gradually increase. By contrast, best experienced payoff-sampling dynamics depend on the difference between two binomial random variables (as each agent compares two samples of size $k$, one when playing $a_i$ and one when playing $b_i$), allowing for bimodal distributions.
\vspace{-5px}
\section{Related Literature}\label{sec:related}
\vspace{-5px}
In this section, we survey the two strands of literature most relevant to our paper: (1) the instability of miscoordination and (2) sampling dynamics.
\paragraph{Instability of Miscoordination}
It is well known that 
strict equilibria satisfy strong
stability 
refinements, while mixed equilibria with miscoordination
do not even meet weak stability refinements. Specifically, strict equilibria satisfy the strong refinement of evolutionary stability (\citealp*{smith1973logic}),
whereas  mixed equilibria fail to satisfy
neutral stability (\citealp*{smith1982evolution}) or the milder
refinement of weak stability (\citealp*{heller2017instability}). Moreover,
it is well known that interior stationary states in 
all multiple-population
games cannot be asymptotically stable under the widely studied replicator
dynamics 
(\citealp[Proposition 2]{ritzberger1995evolutionary}).
Instability of interior
states is further documented for various classes of learning dynamics in
\citet*{crawford1989learning}.

Moreover, various papers in the literature have proven that populations playing coordination games 
converge to one of the pure equilibria from almost all initial states under various dynamics.
\citet*{kaniovski1995learning} show that there is global convergence to one
of the pure equilibria under sampling dynamics with sufficiently
large samples.
\citet*{oprea2011separating} show that there is global convergence to one of the pure equilibria under monotone dynamics (i.e., under any dynamics in which
an action becomes more frequent iff it yields a higher payoff than the alternative action).\footnote{\citet*{oprea2011separating} demonstrated this in the context of hawk--dove games, but the proof can be extended to encompass all coordination games.} 
The stochastic evolutionary dynamics literature (pioneered by \citealp*{kandori1993learning,young1993evolution};
see also the recent application to hawk--dove games in
\citealp*{bilancini2022memory}) shows that only pure equilibria
can be stochastically stable in large finite populations, where agents typically best respond to a large sample but occasionally err.

Conventional wisdom from the above results is that
interior states with miscoordination are unstable and unlikely to result from social learning (as summarized by \citealp{myerson2015tenable}, see Footnote \ref{footnote1}). We show that
this is not true 
under plausible learning dynamics in which agents base
their behavior on sampling actions of other agents, and there is
significant heterogeneity in sample size across the population.

\paragraph{Literature on Sampling Dynamics}
Sampling dynamics 
(aka sampling best-response dynamics or action-sampling dynamics) were introduced by \citet*{sandholm2001almost}
and \citet*{osborne2003sampling}. As argued by \citet*{oyama2015sampling},
the deterministic nature of sampling dynamics implies that when
there is convergence to a stable state 
(which is always the case in our setup; see Part 5 of Fact \ref{fact-w-stable} and Proposition \ref{prop-w-stable}), the convergence is fast.\footnote{Conditions in which stochastic dynamics induce fast convergence
are studied in \citet*{kreindler2013fast} and \citet*{arieli2020speed}.}  
\citet*{heller2018social} studied the conditions on the expected sample
size that induce global convergence for all payoff functions and
all sampling dynamics.
Recently, \cite{danenberg2022representative} introduced  a variant of these dynamics that is based on representative sampling with normal noise.\footnote{\citet*{hauert2018effects} use the term ``sampling dynamics'' to refer to a variant of replicator dynamics in which two agents are likely to be matched if one agent samples and mimics the other agent's behavior. Their conception of sampling dynamics differs from ours and from the literature cited above.} \citet*{arieli2024private} study sampling dynamics in a \textit{social learning} setting and provide conditions under which there is efficient product adoption.

\citet*{salant2020statistical} 
generalized sampling dynamics by allowing agents to use various
procedures to infer  opponents' aggregate behavior from their samples (the analysis is further extended in \citealp{sawa2021statistical};
see also \citealp{patil2024optimal}, who introduce a designer that can endogenize an agent's sample size). \citeauthor*{salant2020statistical} focused on unbiased inference procedures in which the agent's expected
belief about the share of opponents who play 
an action coincides with the
sample mean. 
Examples of unbiased procedures are maximum likelihood
estimation, beta estimation with a prior representing complete ignorance
and a truncated normal posterior around the sample mean. In our setup,
the payoffs are linear in the share of agents who play 
action $a_j$, which
implies that the agent's perceived best response depends only on the
expectation of her posterior belief. 
This implies that the behavior of agents in our model can be interpreted as each new agent utilizing her own sample to calculate  an 
arbitrary unbiased estimation procedure (such as maximum likelihood estimation) for the overall behavior of the opposing population.

\section{Conclusion\label{sec:Conclusion}}
Conventional wisdom, supported by key results in evolutionary game theory, posits that only pure (coordinated) outcomes are reasonable  long-run predictions of behavior in coordination games. 
By contrast, we show that plausible learning dynamics, in which new agents rely on samples to estimate and best respond to the behavior of the opponents' population, can induce stable miscoordination.
This happens if there is heterogeneity in sample size: some agents
have accurate information from large samples of the opponents' aggregate behavior,
while other agents rely on smaller samples.

Our baseline model examines two-player two-action coordination games and employs a specific class of sampling dynamics. Online Appendix \ref{online-appendix}  shows that our qualitative results hold in a much broader setup: more than two players, more than two actions, and under various classes of learning dynamics. 

Our key insight is that heterogeneous noise levels can induce stable miscoordination. Such heterogeneity is realistic in many real-world contexts, like housing markets, where professional investors and inexperienced participants interact. Furthermore, our model’s predictions can be experimentally tested, as discussed in Section \ref{Sec-discussion}.

%\newpage
\appendix
\section{Appendix: Formal Proofs}\label{appendix-proofs}
\subsection{General Result for Binomial Distributions\label{subsec:General-Result-forbinomial}}

Recall our notation both of $X(k,p)\sim Bin\left(k,p\right),$ denoting
a random variable with binomial distribution, and of the function
$F_{m}^{k}\left(p\right)\equiv\Pr\left(X\left(k,p\right)\geq m\right)$.
We next prove a result on the uniqueness of interior fixed points of binomial distributions. This result, which may be of independent interest, plays a crucial role in the proof of Theorem \ref{thm:global-stability-pure}.
\begin{prop}\label{prop-binomial}
Fix arbitrary integers satisfying $0<m_{1}\leq k_{1}$
and $0<m_{2}\leq k_{2}$. There exists at most one $p\in\left(0,1\right)$
such that $\left(F_{m_{1}}^{k_{1}}\circ F_{m_{2}}^{k_{2}}\right)\left(p\right)=p$.
\end{prop}
\setcounter{prop}{3}
\begin{proof}
Let $w_{i}\left(p\right)\equiv F_{m_{i}}^{k_{i}}\left(p\right)$ for
each $i\in\left\{ 1,2\right\} $, $F\left(p\right)\equiv\left(F_{m_{1}}^{k_{1}}\circ F_{m_{2}}^{k_{2}}\right)\left(p\right)\equiv\left(w_{1}\circ w_{2}\right)\left(p\right),$ and
$G\left(p\right)=F\left(p\right)-p$. In what follows, we show that the function $G(\cdot)$ has at most one $p\in\left(0,1\right)$ such that $G\left(p\right)=0,$ which proves the result.

We have $G\left(0\right)=G\left(1\right)=0$. Assume to the contrary
that there exist two different interior points $0<\underline{p}<\overline{p}<1$
such that $G\left(\underline{p}\right)=G\left(\overline{p}\right)$.
Then $G$ equals zero at four points in the interval $\left[0,1\right]$.
By Rolle's theorem, this implies that $G'$ is equal to zero in at least
three points in the interval $\left(0,1\right),$ which further
implies that $G''$ is equal to zero in at least two interior points
in the interval $\left(0,1\right)$. Observe that $G''\equiv F''$.
Thus, in order to obtain a contradiction, we have to show that $F''\left(p\right)=0$
in at most one interior point. Recall that (see, e.g., \citealp[Eq. (5)]{green1983fixed})
\vspace{-10px}\begin{equation}\label{eq:green1983_eqn}
w'_{i}\left(p\right)=m_{i}\left(\begin{array}{c}
k_{i}\\
m_{i}
\end{array}\right)p^{m_{i}-1}\left(1-p\right)^{k_{i}-m_{i}}\,\Rightarrow\,\frac{w''_{i}\left(p\right)}{w'_{i}\left(p\right)}=\frac{m_{i}-1}{p}-\frac{k_{i}-m_{i}}{1-p}.
\end{equation}
For $p\in (0,1),$ using Eq. \eqref{eq:green1983_eqn}, we compute as follows:
\vspace{-10px}
\begin{align*}
F'\left(p\right) &= w_{1}'\left(w_{2}\left(p\right)\right)w_{2}'\left(p\right)\\
F''\left(p\right)&=w_{1}''\left(w_{2}\left(p\right)\right)\left(w_{2}'\left(p\right)\right)^{2}+w_{1}'\left(w_{2}\left(p\right)\right)w_{2}''\left(p\right)\\
&= w_{1}'\left(w_{2}\left(p\right)\right)w_{2}'\left(p\right)\left[\left(\frac{m_{1}-1}{w_{2}\left(p\right)}-\frac{k_{1}-m_{1}}{1-w_{2}\left(p\right)}\right)w_{2}'\left(p\right)+\frac{m_{2}-1}{p}-\frac{k_{2}-m_{2}}{1-p}\right].
\end{align*}

The fact that each $w_{i}\left(p\right)$ is strictly increasing implies
that $F''\left(p\right)=0$ iff 
\[
\left(\frac{m_{1}-1}{w_{2}\left(p\right)}-\frac{k_{1}-m_{1}}{1-w_{2}\left(p\right)}\right)w_{2}'\left(p\right)=\frac{k_{2}-m_{2}}{1-p}-\frac{m_{2}-1}{p}\Leftrightarrow
\]
\[
m_{2}\left(\begin{array}{c}
k_{2}\\
m_{2}
\end{array}\right)\left(\frac{m_{1}-1}{w_{2}\left(p\right)}-\frac{k_{1}-m_{1}}{1-w_{2}\left(p\right)}\right)p^{m_{2}-1}\left(1-p\right)^{k_{2}-m_{2}}=\frac{k_{2}-m_{2}}{1-p}-\frac{m_{2}-1}{p}\Leftrightarrow
\]
{\footnotesize{}
\[
\!\!\!\!\!\!\,\,\!\!m_{2}\left(\begin{array}{c}
k_{2}\\
m_{2}
\end{array}\right)\left(\frac{\left(m_{1}-1\right)p^{m_{2}}\left(1-p\right)^{k_{2}-m_{2}}}{\sum_{l=m_{2}}^{k_{2}}\left(\begin{array}{c}
k_{2}\\
l
\end{array}\right)p^{l}\left(1-p\right)^{k_{2}-l}}-\frac{\left(k_{1}-m_{1}\right)p^{m_{2}}\left(1-p\right)^{k_{2}-m_{2}}}{\sum_{l=0}^{m_{2}-1}\left(\begin{array}{c}
k_{2}\\
l
\end{array}\right)p^{l}\left(1-p\right)^{k_{2}-l}}\right)\frac{1}{p}=\frac{k_{2}-m_{2}}{1-p}-\frac{m_{2}-1}{p}\Leftrightarrow
\]
}
\vspace{-20px}
\[
m_{2}\left(\begin{array}{c}
k_{2}\\
m_{2}
\end{array}\right)\left(\frac{\left(m_{1}-1\right)}{\sum_{l=m_{2}}^{k_{2}}\left(\begin{array}{c}
k_{2}\\
l
\end{array}\right)\left(\frac{p}{1-p}\right)^{l-m_{2}}}-\frac{k_{1}-m_{1}}{\sum_{l=0}^{m_{2}-1}\left(\begin{array}{c}
k_{2}\\
l
\end{array}\right)\left(\frac{1-p}{p}\right)^{m_{2}-l}}\right)\frac{1}{p}=\frac{k_{2}-m_{2}}{1-p}-\frac{m_{2}-1}{p}.
\]
One can verify that the left-hand side of the above equation is strictly decreasing and the right-hand side is strictly increasing in $p.$ Therefore, there can be at most one point $p^{\ast}\in (0,1)$ where $F''(p^{\ast})=0.$ This completes the proof.\qedhere
\end{proof}

\subsection{Standard Definitions of Dynamic Stability\label{sub-standard-definitions}}

For completeness, we present some standard definitions
of dynamic stability that are used in the paper (see, e.g., \citealp[Chapter 5]{weibull1997evolutionary}).

A state is said to be stationary if it is a rest point of the dynamics. 
\begin{defn}
\label{def:equilbirium}\textit{One-population dynamics}:
$p^{*}\in\left[0,1\right]$
is \emph{stationary} if $w(p^{*})=p^{*}$.\\
\textit{Two-population dynamics}: $\mathbf{p}^{*}\in\left[0,1\right]^{2}$
is \emph{stationary} if $w_{i}(\mathbf{p}^{*})=p_{i}^{*}$
for each $i\in\left\{ 1,2\right\} $.
\end{defn}
A stationary state is Lyapunov stable if a population beginning near $\mathbf{p}^{*}$ stays close to it. $\mathbf{p}^{*}$ is asymptotically stable if it is Lyapunov stable and, in addition, nearby states eventually converge to it. A state is unstable if it is not Lyapunov stable. Formally (where for two-population dynamics the states should be boldfaced):
\begin{defn}
\label{def:lyaponouv-stability}
A stationary state $p^{*}$
is \emph{Lyapunov stable} if for every neighborhood $U$ of $p^{*}$
there is a neighborhood $V\subseteq U$ of $p^{*}$ such
that if the initial state $p\left(0\right)\in V$, then $p\left(t\right)\in U$
for all $t>0$. A state is \emph{unstable} if it is not Lyapunov stable. 
\end{defn}
\begin{defn}
A stationary state $p^{*}$ is \emph{asymptotically
stable} (or \emph{locally stable}) if it is Lyapunov stable and there
is some neighborhood $U$ of $p^{*}$ such that all trajectories
initially in $U$ converge to $p^{*}$; i.e., $p\left(0\right)\in U$
implies that $\lim_{t\rightarrow\infty}p\left(t\right)=p^{*}$. 
\end{defn}

\subsection{Proof of Proposition \ref{prop-w-stable} (Asymptotically Stable States)}\label{app-proof-of-Prop-1} 
The following three lemmas are helpful in proving  Proposition \ref{prop-w-stable} (and later results).
Our first lemma shows that any trajectory reaches the area between the curves.
\begin{lem}\label{lem:conv-to-area-between-curves}
Either $\lim_{t\rightarrow\infty}\mathbf{p}\left(t\right)$ exists and is a stationary state, or there exists
$t<\infty$ such that either $w_{1}^{-1}\left(p_{1}\left(t\right)\right)\leq p_{2}\left(t\right)\leq w_{2}\left(p_{1}\left(t\right)\right)$
or $w_{2}\left(p_{1}\left(t\right)\right)\leq p_{2}\left(t\right)\leq w_{1}^{-1}\left(p_{1}\left(t\right)\right)$.
\end{lem}
\begin{proof}
We say that state $\textbf{p}$
is \emph{above} (resp., \emph{below}) curve $w_{2}(p_{1})$ if $p_{2}>w_{2}(p_{1})$
(resp., ($p_{2}<w_{2}(p_{1})$). Similarly, we say that state $\textbf{p}$
is to the \emph{right} (resp., \emph{left}) of the curve $w_{1}(p_{2})$
if $p_{1}>w_{1}(p_{2})$ (resp., $p_{1}<w_{1}(p_{2})$). We say that
state $\textbf{p}$ is on the curve $w_{i}(p_{j})$ if $p_{i}=w_{i}(p_{j}).$
Since both curves are strictly increasing, we define being above a curve as being to the left of it, and being below a curve as being to the right of it. The states on the curve $w_{2}(p_{1})$
(resp., $w_{1}(p_{2})$) are characterized by having $\dot{p_{2}}=0$
(resp., $\dot{p_{1}}=0$). Observe that $\dot{p_{2}}>0$ (resp., $\dot{p_{2}}<0$)
in any state $\textbf{p}$ above and to the left (resp., below and
to the right) of the curve $w_{2}(p_{1})$. Similarly, $\dot{p_{1}}>0$
(resp., $\dot{p_{1}}<0$) in any state $\textbf{p}$ above and to
the left (resp., below and to the right) of the curve $w_{1}(p_{2})$. 

Any state $\textbf{p}\in[0,1]$ can be classified in one of $9=3\cdot3$
classes, depending on its relative location with respect to the two
curves, i.e., whether $\textbf{p}$ is below, above, or on each of
the two curves $w_{i}(p_{j})$. If state $\textbf{p}$ is on (resp.,
above, below) the curve $w_{2}(p_{1})$, then $\dot{p_{2}}$ is zero
(resp., negative, positive). Similarly, if state $\textbf{p}$ is
on (resp., above, below) the curve $w_{1}(p_{2})$, then $\dot{p_{1}}$
is zero (resp., positive, negative).
In particular, any state $\textbf{p}$ that is above (resp., below)
both curves must satisfy $\dot{p_{1}}>0>\dot{p_{2}}$. This implies
that any trajectory that begins above (resp., below) both curves must
always move downward and to the right. 
This (together with the fact
that $\dot{\boldsymbol{p}}\left(t\right)=(0,0)$  only at the intersection points of the two curves)
implies that the trajectory must either converge to a stationary state (i.e., to an intersection point of the two curves) or cross one of the curves, and reach 
a state $\boldsymbol{p}\left(t\right)$ that satisfies either $w_{1}^{-1}\left(p_{1}\left(t\right)\right)\leq p_{2}\left(t\right)\leq w_{2}\left(p_{1}\left(t\right)\right)$
or $w_{2}\left(p_{1}\left(t\right)\right)\leq p_{1}\left(t\right)\leq w_{1}^{-1}\left(p_{1}\left(t\right)\right)$.
\end{proof}
Our second lemma shows that any trajectory that reaches  the area between the curves converges to a neighboring stationary state.
\begin{lem}
\label{lem:convergence-to-which-neighbour}Let $\boldsymbol{p}\left(t\right)$
be an interior state for some $t\geq0$ and let $\underline{\boldsymbol{p}},\boldsymbol{\overline{p}}$ be neighboring stationary states with $\underline{p}_{i}\leq p_i(t)\leq \overline{p}_{i}$. We have
\begin{enumerate}
\item If $p_{2}\left(t\right)\in\left[w_{1}^{-1}\left(p_{1}\left(t\right)\right),w_{2}\left(p_{1}\left(t\right)\right)\right]$,
then $\lim_{t\rightarrow\infty}\mathbf{p}\left(t\right)=\overline{\boldsymbol{p}}$,
and
\item If $p_{2}\left(t\right)\in\left[w_{2}\left(p_{1}\left(t\right)\right),w_{1}^{-1}\left(p_{1}\left(t\right)\right)\right]$,
then $\lim_{t\rightarrow\infty}\mathbf{p}\left(t\right)=\underline{\boldsymbol{p}}$.
\end{enumerate}
\end{lem}

\begin{proof}
We first show why we can assume w.l.o.g. that $\boldsymbol{p}\left(t\right)$
is strictly between the two curves. If $\boldsymbol{p}\left(t\right)$
crosses one of the curves and is strictly above (resp., below) the
other curve, then the dynamics must take the populations to a
state that is strictly below one of the curves and strictly above
the other curve. This is so because at the intersection point one
of the $\dot{p_{i}}$'s is zero and the other derivative $\dot{p}_j$
is negative (resp., positive), which implies that the dynamics take
the trajectory below and to the right (resp., above and to the left)
of the curve that was crossed.

\noindent Next assume that $p_{i}\left(t\right)\in\left(w_{1}^{-1}\left(p_{1}\left(t\right)\right),w_{2}\left(p_{1}\left(t\right)\right)\right)$
(resp., $p_{i}\left(t\right)\in\left(w_{2}\left(p_{1}\left(t\right)\right),w_{1}^{-1}\left(p_{1}\left(t\right)\right)\right).$
By the classification presented in the proof of Lemma \ref{lem:conv-to-area-between-curves},
the trajectory must move upward and to the right, i.e., $\dot{p_{1}},\dot{p_{2}}>0$
(resp., downward and to the left, i.e., $\dot{p_{1}},\dot{p_{2}}<0$).
This implies that the trajectory must cross one of the curves. The
intersection point cannot be only on the curve of $w_{2}(p_{1})$ (resp.,
$w_{1}(p_{2})$), because at such a point the trajectory moves horizontally
to the left (vertically upward), which implies that it must cross
the lower (resp., higher) curve $w_{2}(p_{1})$ (resp.,$w_{1}(p_{2})$)
from the left side (resp., from below), and we get a contradiction.
This implies that the intersection point must be the closest intersection
point of the two curves to the right (resp., left) of $p_{i}\left(t\right)$,
namely, $\overline{\boldsymbol{p}}$ (resp., $\boldsymbol{\underline{p}}$).
\end{proof}

\begin{lem}\label{lem-stationary-char}
Let $\hat{\boldsymbol{p}}$
be a stationary state. $\hat{\boldsymbol{p}}$ is asymptotically stable if both of the following conditions hold:
\begin{enumerate}
\item \textbf{Left neighborhood}: If $\hat{p}_{1}>0$ then there exists
$\underline{p}_{1}\in\left(0,\hat{p}_{1}\right)$ such that $w_{2}\left(p_{1}\right)>w_{1}^{-1}\left(p_{1}\right)$
for any $p_{1}\in\left(\underline{p}_{1},\hat{p}_{1}\right)$, and
\item \textbf{Right neighborhood}: If $\hat{p}_{1}<1,$ then there exists
$\overline{p}_{1}\in\left(\hat{p}_{1},1\right)$ such that $w_{2}\left(p_{1}\right)<w_{1}^{-1}\left(p_{1}\right)$
for any $p_{1}\in\left(\hat{p}_{1},\overline{p}_{1}\right)$.
\end{enumerate}
Moreover, if either of the above two conditions is not satisfied, then
$\hat{\boldsymbol{p}}$ is unstable.
\end{lem}
\begin{proof}
Assume that Conditions 1 and 2 hold. Let $\boldsymbol{p}\neq\hat{\boldsymbol{p}}$
be any sufficiently close state. By Lemma \ref{lem:convergence-to-which-neighbour},
any trajectory beginning at $\boldsymbol{p}$ will enter one of the
two areas between the two curves on either side of $\hat{\boldsymbol{p}}$.
By Lemma \ref{lem:convergence-to-which-neighbour}, Condition
1 (resp., 2) implies convergence to $\hat{\boldsymbol{p}}$ if
the trajectory has entered the area between the curves to the left
(resp., right) of $\hat{\boldsymbol{p}}.$ This implies that any trajectory
that starts sufficiently close to $\hat{\boldsymbol{p}}$ must converge
to $\hat{\boldsymbol{p}}$, and, thus, $\hat{\boldsymbol{p}}$ is asymptotically
stable.

Next assume that $\hat{p}_{1}>0$ (resp., $\hat{p}_{1}<1$) and Condition 1 (resp., Condition 2) is not satisfied. This implies that $w_{2}\left(p_{1}\right)<w_{1}^{-1}\left(p_{1}\right)$
(resp., $w_{2}\left(p_{1}\right)>w_{1}^{-1}\left(p_{1}\right)$) for
any $p_{1}$ that is sufficiently close to $\hat{p}_{1}$ from the left (resp., right). By Lemma \ref{lem:convergence-to-which-neighbour},
this implies that a trajectory starting at $\left(p_{1},p_{2}\right)$
with $p_{1}$ sufficiently close to $\hat{p}_{1}$ from the left (resp., right)  and with $p_{2}\in\left(w_{2}\left(p_{1}\right),w_{1}^{-1}\left(p_{1}\right)\right)$
(resp., $p_{2}\in\left(w_{1}^{-1}\left(p_{1}\right),w_{2}\left(p_{1}\right)\right)$)
converges to the neighboring stationary point on the left (resp., right) side
of $\hat{\boldsymbol{p}}$. Thus, $\hat{\boldsymbol{p}}$ is unstable.
\end{proof}

Parts 1--3 of Proposition \ref{prop-w-stable} are immediately implied by Lemma \ref{lem-stationary-char}. To prove Part 4 assume w.l.o.g. that $\underline{p}_{1}<\overline{p}_{1}$.
By Lemma \ref{lem:convergence-to-which-neighbour}, the fact
that $\underline{\boldsymbol{p}}$ is asymptotically stable implies
that $w_{2}\left(p_{1}\right)>w_{1}^{-1}\left(p_{1}\right)$ for any
$p_{1}\in\left(\underline{p}_{1},\overline{p}_{1}\right),$ which, in turn implies that $\overline{\boldsymbol{p}}$ is unstable. Part 5 is implied by Lemmas \ref{lem:conv-to-area-between-curves} and \ref{lem:convergence-to-which-neighbour}. 

\subsection {Proof of Proposition \ref{prop:stationary-coincide} (Symmetric Environments)}\label{proof-prop-2-sym}

We first show that in symmetric coordination games, the set of stationary states coincide under both one-population dynamics and two-population dynamics. 
\begin{lem}\label{lemma:stationary-coincide}
    Let $(u,\theta)$ be a symmetric environment.  State $\textbf{p}$ is  stationary  under two-population dynamics iff the state is symmetric (i.e., $p\equiv p_1=p_2$), and $p$ is  stationary under one-population dynamics.
\end{lem}
\begin{proof}
    Assume that $\textbf{p}$ is stationary under two-population dynamics. Assume to the contrary that $p_1\neq p_2$. Assume w.l.o.g. that $p_1<p_2$. Let $w\equiv w_1=w_2$ represent the symmetric learning dynamics. Since $w$ is strictly increasing, we have $p_2=w(p_1)<w(p_2)=p_1$, a contradiction.Therefore, all stationary states under two-population dynamics are symmetric (i.e., of the form $(p,p)$). Observe that a symmetric state $(p,p)$ is stationary under two-population dynamics iff $p$ is stationary under one-population dynamics, because stationarity under both dynamics is characterized by the equation $p=w(p)$.
\end{proof}

Lemma \ref{lemma:stationary-coincide} implies that all asymptotically stable states under two-population dynamics are symmetric. The strict monotonicity of $w\equiv w_1=w_2$ implies that state $p$ satisfies Conditions 1--3 in Fact \ref{fact-w-stable} iff the symmetric state $(p,p)$ satisfies Conditions 1--3 in Propostion \ref{prop-w-stable}. This, in turn, implies that state $p$ is asymptotically stable under one-population dynamics iff  $(p,p)$ is asymptotically stable under two-population dynamics.

\subsection{Proof of Theorem \ref{thm:global-stability-pure} (Homogeneous Populations)}\label{app-proof-of-thm1'}

The fact that $\theta_{i}\equiv k_{i}>1$ implies that $w_{i}\left(p\right)=F_{m_{i}}^{k_{i}}\left(p\right)$
for some $1\leq m_{i}\leq k_{i}$. This implies that any stationary
state $\boldsymbol{p}^{\ast}$ must satisfy $\left(F_{m_{1}}^{k_{1}}\circ F_{m_{2}}^{k_{2}}\right)\left(p_{1}^{\ast}\right)=p_{1}^{\ast}$.
Proposition \ref{prop-binomial} implies that this holds for at most
one interior state $\hat{\boldsymbol{p}}$. Therefore, the stationary state $\hat{\boldsymbol{p}}$ (if it exists) is a neighbor of both pure stationary states $\left(0,0\right)$ and $\left(1,1\right)$.
Part 2 of Proposition \ref{prop:pure-equilbiria-satbility} implies that state $\textbf{a}$ is asymptotically stable (because $\theta_1(1)=0$). Finally, Part 4 of Proposition  \ref{prop-w-stable}
implies that $\hat{\boldsymbol{p}}$ is unstable. \qedhere 

\subsection{Proof of Theorem \ref{thm:locally-stable_interior_if_pure_is_unstable} (Stable Miscoordination, $u_1,u_2>1$)}\label{Proof-Thm-2'}

We begin the proof by presenting some notation.
Let 
%$w_{i}^{k_{i}}\left(p_{j}\right)$ (resp.,
$w_{i}^{\theta_{i}}\left(p_{j}\right)$
denote the sampling dynamics induced by the distribution $\theta_{i}$ of sample sizes.
Given a distribution $\theta$ of sample sizes, sample size  $\bar{k}\notin \textrm{supp}(\theta),$ and mass $\alpha \in(0,1)$, let $\theta\alpha\bar{k}$ be the mixed sample size distribution that assigns mass $\alpha$ to distribution  $\theta$ and mass $1-\alpha$ to  $\bar{k}$. That is, 
\vspace{-5pt}
\[
\left(\theta\alpha\bar{k}\right)(k)=\begin{cases}
\alpha\cdot\theta\left(k\right) & k\neq\bar{k}\\
\left(1-\alpha\right) & k=\bar{k}.
\end{cases}
\]
\vspace{-5pt}
Let $\hat k_{i}$ be the maximal sample size in the support of $\theta_i$ for each player $i$, i.e., $\hat k_{i}=\max\left(\textrm{supp}(\theta_{i})\right)\leq u_i+1$. In what follows, we show that there exist $\alpha_1,\alpha_2\in(0,1)$, such that the environment $\left(\boldsymbol{u},\left(\theta_{1}\alpha_{1}\bar{k},\theta _{2}\alpha_{2}\bar{k}\right)\right)$
admits an asymptotically stable interior state $\boldsymbol{\hat{p}}$ for any   sufficiently large $\bar{k}$. 
 The fact that  $\hat k_{i}\leq u_i+1$ implies that

\vspace{-15pt}
\[
w_{i}^{\theta_{i}}\left(p_{j}\right)=\sum_{k_{i}}\theta_{i}\left(k_{i}\right)\Pr\left(X\left(k_{i},p_{j}\right)\geq1\right)=
\]
\vspace{-15pt}
\[
1-\sum_{k_{i}}\theta_{i}\left(k_{i}\right)\left(1-p_{j}\right)^{k_{i}}=\sum_{k_{i}}\theta_{i}\left(k_{i}\right)\left(k_{i}p_{j}-\left(\begin{array}{c}
k_{i}\\
2
\end{array}\right)\left(p_{j}\right)^{2}\right)+f\left(p_{j}\right),
\]
where $f(p_j)=O\left(\left(p_{j}\right)^{3}\right).$
Fix a sufficiently small $\epsilon>0$. Let $\hat{p}\in\left(0,p_{1}^{NE}-\epsilon\right)$
be sufficiently small such that $f(p_j)<\epsilon$
for any $p_{j}<\hat{p}$. 
Observe that $\sum_{k_{i}}\theta_{i}\left(k_{i}\right)k_{i}=\mathbb{E}(\theta_i)>1.$ 
Observe that for each $i\in\left\{ 1,2\right\} $,
there exists an interval of $\alpha_i$'s in $(0,1)$ that satisfy (1) $\alpha_{i}\sum_{k_{i}}\theta_{i}\left(k_{i}\right)k_{i}>1$
and 
(2) $\alpha_{i}\sum_{k_{i}}\theta_{i}\left(k_{i}\right)\left(k_{i}-\left(\begin{array}{c}
k_{i}\\
2
\end{array}\right)\hat{p}\right)<1-2\epsilon$. Observe that (1) implies that $w_{i}^{\theta_i}\left(p_{j}\right)>\frac{p_{j}}{\alpha_{i}}$
in a right neighborhood of zero, and that (2) implies that $w_{i}^{\theta_{i}}\left(p_{j}\right)<\frac{p_{j}}{\alpha_{i}}$ in a left neighborhood of $\hat{p}$. Further observe that $\lim_{k\rightarrow\infty}w_{2}^{k}\left(\hat{p}\right)=0$
and $\lim_{k\rightarrow\infty}\left(w_{1}^{k}\right)^{-1}\left(\hat{p}\right)=p_{2}^{NE}$.
This implies that there exists sufficiently large $\overline{k}$ such that $w_{2}^{\overline{k}}\left(\hat{p}_{1}\right)<\epsilon$
and $\left(w_{1}^{\overline{k}}\right)^{-1}\left(\hat{p}\right)>p_{2}^{NE}-\epsilon>\hat{p}$.

Observe that $w_{i}^{\theta_{i}\alpha_{i}\overline{k}}\left(p_{j}\right)>p_{j}$
in a right neighborhood of zero, and $w_{i}^{\theta_{i}\alpha_{i}\overline{k}}\left(p_{j}\right)<p_{j}$
in a left neighborhood of $\hat{p}$. This, in turn, implies that
$w_{2}^{\theta_{2}\alpha_{2}\overline{k}}\left(p_{1}\right)>p_{1}>\left(w_{1}^{\theta_{1}\alpha_{1}\overline{k}}\right)^{-1}\left(p_{1}\right)$
in a right neighborhood of zero, and $w_{2}^{\theta_{2}\alpha_{2}\overline{k}}\left(p_{1}\right)<p_{1}<\left(w_{1}^{\theta_{1}\alpha_{1}\overline{k}}\right)^{-1}\left(p_{1}\right)$
in a left neighborhood of $\hat{p}$. Thus, there exists a stationary
state $\tilde{\boldsymbol{p}}$ that satisfies $0<\tilde{p}_{1}<\hat{p}<1$,
and $w_{2}^{\theta_{i}\alpha_{i}\overline{k}}\left(p_{1}\right)>\left(w_{1}^{\theta_{i}\alpha_{i}\overline{k}}\right)^{-1}\left(p_{1}\right)$
(resp., $w_{2}^{\theta_{i}\alpha_{i}\overline{k}}\left(p_{1}\right)<\left(w_{1}^{\theta_{i}\alpha_{i}\overline{k}}\right)^{-1}\left(p_{1}\right)$)
in a left (resp., right) neighborhood of $\tilde{p}_{1}$. This implies,
by Part 3 of Proposition \ref{prop-w-stable}, that
$\tilde{\boldsymbol{p}}$ is asymptotically stable.  The above argument is illustrated in the left panel of Fig. \ref{fig:fig3-local-stablity}.

\subsection{Proof of Theorem \ref{thm:locally-stableinterior_ASD} (Stable
Miscoordination > 50\%)\label{subsec:Proof-of-Theorem-3}}
We use the notation introduced in  the proof of Theorem \ref{thm:locally-stable_interior_if_pure_is_unstable}. Observe that
\vspace{-15pt}
\[\vspace{-5pt}
w_{1}^{\theta_{1}}\left(p_{2}\right)=\Pr\left(\textrm{a new agent in Pop. 1 has at least 1 }a_{2}\textrm{ in her sample}\right)=1-\sum_{k_{1}}\theta_{1}\left(k_{1}\right)(1-p_{2})^{k_{1}}
\vspace{-15pt}
\]
\vspace{-5pt}
is the same for all values of $u_{1}>\hat k_{1}\equiv\max\left(\textrm{supp}(\theta_{i})\right)$, and, similarly, 
\vspace{-10pt}
\[
w_{2}^{\theta_{2}}\left(p_{1}\right)=\Pr\left(\textrm{a new agent in Pop. 2 has only }a_{1}\textrm{'s in her sample}\right)=\sum_{k_{2}}\theta_{2}\left(k_{2}\right)(p_{1})^{k_{2}}
\vspace{-15pt}
\]
\noindent
is the same for all values of $u_{2}<\frac{1}{\hat k_{2}}$. Fix a sufficiently
small $\epsilon>0$. 
Let $\bar{u}_{1}>\hat k_{1},\bar{u}_{2}<\frac{1}{\hat k_{2}}$ be payoffs satisfying
$p_{2}^{NE}=\frac{1}{1+\bar{u}_{1}}<\frac{1}{2^{\hat k_{2}+1}}-\epsilon$
and $p_{1}^{NE}=\frac{1}{1+\bar{u}_{2}}>1-\frac{1}{2^{\hat k_{1}+1}}+\epsilon$.
Fix any
payoff profile $\textbf{u}$ satisfying $u_{1}>\bar{u}_{1}$ and $u_{2}<\bar{u}_{2}$. Let $\overline{k}$
be sufficiently large such that $w_{i}^{\overline{k}}\left(p_{j}^{NE}-\epsilon\right)<\epsilon$,
$w_{i}^{\overline{k}}\left(p_{j}^{NE}+\epsilon\right)>1-\epsilon$.
In what follows, we show that  $\left(\textbf{u},(\theta_{1}\frac{1}{2}\bar{k},\theta_{2}\frac{1}{2}\bar{k})\right)$
admits an asymptotically stable interior state $\hat{\boldsymbol{p}}$
that satisfies $\hat{p}_{2}<\frac{1}{2}<\hat{p}_{1}$. 
Observe that: 
\begin{equation}
\frac{1}{2}>w_{2}^{\theta_{2}\frac{1}{2}\bar{k}}\left(\frac{1}{2}\right)=\frac{1}{2}\left(w_{2}^{\theta_{2}}\left(\frac{1}{2}\right)+w_{2}^{\bar{k}}\left(\frac{1}{2}\right)\right)>\frac{1}{2^{\hat{k}_{2}+1}}+\frac{1}{2}w_{2}^{\bar{k}}\left(\frac{1}{2}\right)>p_{2}^{NE}+\epsilon>\left(w_{1}^{\theta_{1}}\right)^{-1}\left(\frac{1}{2}\right),
\label{eq:w_2^=00005Ctheta_2}
\end{equation}
\begin{equation}
\frac{1}{2}<w_{1}^{\theta_{1}\frac{1}{2}\bar{k}}\left(\frac{1}{2}\right)<\frac{1}{2}\left(1-\frac{1}{2^{\hat k_{1}}}+w_{1}^{\bar{k}}\left(\frac{1}{2}\right)\right)<1-\frac{1}{2^{\hat k_{1}+1}}<p_{1}^{NE}-\epsilon<\left(w_{2}^{\theta_{2}}\right)^{-1}\left(\frac{1}{2}\right).\label{eq:w_2^=00005Ctheta_2-1}
\end{equation}
These inequalities imply that the curve $w_{2}^{\theta_{2}\frac{1}{2}\bar{k}}\left(p_{1}\right)$
is (1) above the curve $\left(w_{1}^{\theta_{1}\frac{1}{2}\bar{k}}\right)^{-1}\left(p_{1}\right)$
at $p_{1}=\frac{1}{2}$, and (2) below $\left(w_{1}^{\theta_{1}\frac{1}{2}\bar{k}}\right)^{-1}\left(p_{1}\right)$
at $w_{1}^{\theta_{1}\frac{1}{2}\bar{k}}\left(\frac{1}{2}\right)>\frac{1}{2}$ 
(recall that an increasing curve is below another increasing curve iff it is to the right of the other curve). This
implies that there is a stationary state $\hat{\boldsymbol{p}}$ between
$\frac{1}{2}$ and $w_{1}^{\theta_{1}\frac{1}{2}\bar{k}}\left(\frac{1}{2}\right)$ such
that the curve $w_{2}^{\theta_{2}\frac{1}{2}\bar{k}}\left(p_{1}\right)$ is above (resp.,
below) the curve $\left(w_{1}^{\theta_{1}\frac{1}{2}\bar{k}}\right)^{-1}\left(p_{1}\right)$
in a left (resp., right) neighborhood of $\hat{\boldsymbol{p}}$.
Part 3 of Proposition \ref{prop-w-stable}, then, implies
that $\hat{\boldsymbol{p}}$ is asymptotically stable. Note that
$\hat{p}_{2}<\frac{1}{2}<\hat{p}_{1}$, implying that the  miscoordination probability induced by $\hat{\boldsymbol{p}}$ is >50\% because
\begin{align*}
    Pr\left(\textrm{miscoordination}\right)&=p_{1}\left(1-p_{2}\right)+p_{2}\left(1-p_{1}\right)=p_{1}-2p_{1}p_{2}+p_{2}\\
    &=p_{1}-p_{2}\left(2p_{1}-1\right)>p_{1}-\frac{1}{2}\left(2p_{1}-1\right)=\frac{1}{2}.
\end{align*}

\subsection{Proof of Theorem \ref{thm4} (Globally Stable Miscoordination)}\label{subsec:Proof-of-Theorem-4}

The following proposition, which may be of independent interest, characterizes the asymptotic stability of pure equilibria and is key to proving Theorem \ref{thm4}.

\begin{prop}
~\label{prop:pure-equilbiria-satbility} \label{Prop-1}
\vspace{-10pt}
\begin{enumerate}
\item $\mathbb{E}_{<\frac{1}{u_{1}}+1}\left(\theta_{1}\right)\cdot\mathbb{E}_{<\frac{1}{u_{2}}+1}\left(\theta_{2}\right)=\theta_{1}\left(1\right)\cdot\mathbb{E}_{<\frac{1}{u_{2}}+1}\left(\theta_{2}\right)>1\,\Rightarrow$
\textbf{a}=$\left(a_{1},a_{2}\right)$ is unstable; 
\item $\mathbb{E}_{<\frac{1}{u_{1}}+1}\left(\theta_{1}\right)\cdot\mathbb{E}_{<\frac{1}{u_{2}}+1}\left(\theta_{2}\right)=\theta_{1}\left(1\right)\cdot\mathbb{E}_{<\frac{1}{u_{2}}+1}\left(\theta_{2}\right)<1\,\Rightarrow$ $\boldsymbol{a}$
is asymptotically stable;
\item $\mathbb{E}_{\leq u_{1}+1}\left(\theta_{1}\right)\cdot\mathbb{E}_{\leq u_{2}+1}\left(\theta_{2}\right)>1\,\Rightarrow$
\textbf{b}=$\left(b_{1},b_{2}\right)$ is unstable; and
\item $\mathbb{E}_{\leq u_{1}+1}\left(\theta_{1}\right)\cdot\mathbb{E}_{\leq u_{2}+1}\left(\theta_{2}\right)<1\,\Rightarrow$ \textbf{b}=$\left(b_{1},b_{2}\right)$ is asymptotically stable.
\end{enumerate}
\end{prop}
\begin{proof}
    
We are interested in deriving conditions for the stability of pure stationary states. In what follows, we compute the Jacobian of sampling dynamics in the pure state $\boldsymbol{a}=\left(0,0\right)$
(resp., $\boldsymbol{b}=\left(1,1\right)$). To do so, we consider
a slightly perturbed state with a ``very small'' $\epsilon_{i}$
share of agents playing $b_{i}$ (resp., $a_{i})$ in population $i$.
By ``very small,'' we mean that higher-order terms of $\epsilon_{i}$
and $\epsilon_{j}$ are neglected.

Consider a new agent in population $i$ who has a sample size of $k_{i}.$
Action $b_{i}$ (resp., $a_{i}$) has a weakly (resp., strictly) higher
mean payoff against a sample size of $k_{i}$ iff (we neglect the rare cases where an agent has of having multiple $b_{j}$'s (resp., $a_{j}$'s) in her sample).
(1) the sample includes the opponent's action $b_{j}$ (resp.,
$a_{j}$), and (2) $k_{i}<\frac{1}{u_{i}}+1$ (resp., $k_{i}\leq u_{i}+1$).
The probability of (1) is $k_{i}\cdot\epsilon_{j}+o(\epsilon_{j})$,
where $o(\epsilon_{j})$ denotes terms that are 
asymptotically negligible compared to $\epsilon_{j}$,
and, thus, it does not affect the Jacobian as $\epsilon_{j}\rightarrow0$.
This implies that the probability that a new agent in population $i$
(with a random sample size distributed according to $\theta_{i}$) who has a higher mean payoff for action $b_{i}$ (resp., $a_{i}$) against
her sample is $w_{i}(1-\epsilon_{j})=\epsilon_{j}\cdot\mathbb{E}_{<\frac{1}{u_{i}}+1}\left(\theta_{i}\right)+o(\epsilon_{j})$
(resp., $w_{i}(\epsilon_{j})=\epsilon_{j}\cdot\mathbb{E}_{\leq u_{i}+1}\left(\theta_{i}\right)+o(\epsilon_{j})$). Therefore, the sampling dynamics at $(\epsilon_{1},\epsilon_{2})$
(resp., $(1-\epsilon_{1},1-\epsilon_{2})$) can be written as follows
(ignoring the higher-order terms of $\epsilon_{1}$ and $\epsilon_{2}$):
\vspace{-10px}
\begin{equation}
\vspace{-10px}
\dot{\epsilon}_{i}=\epsilon_{j}\cdot\mathbb{E}_{<\frac{1}{u_{i}}+1}\left(\theta_{i}\right)-\epsilon_{i}\,\,\,(\textrm{resp.,}\,\dot{\epsilon}_{i}=\epsilon_{i}-\epsilon_{j}\cdot\mathbb{E}_{\leq u_{i}+1}\left(\theta_{i}\right)).\label{eq:a_b_theta}
\end{equation}
Define: $a_{\theta_{i}}=\mathbb{E}_{<\frac{1}{u_{i}}+1}\left(\theta_{i}\right)$
(resp., $b_{\theta_{i}}=\mathbb{E}_{\leq u_{i}+1}\left(\theta_{i}\right)$).
The Jacobian of the above system of Equations \eqref{eq:a_b_theta}
is then given by $J_{a}=\left(\begin{array}{cc}
-1 & a_{\theta_{1}}\\
a_{\theta_{2}} & -1
\end{array}\right)$ (resp., $J_{a}=\left(\begin{array}{cc}
1 & -b_{\theta_{1}}\\
-b_{\theta_{2}} & 1
\end{array}\right)$). The eigenvalues of $J_{a}$ (resp., $J_{b}$) are $-1-\sqrt{a_{\theta_{1}}a_{\theta_{2}}}$
and $-1+\sqrt{a_{\theta_{1}}a_{\theta_{2}}}$ (resp., $-1-\sqrt{b_{\theta_{1}}b_{\theta_{2}}}$
and $-1+\sqrt{b_{\theta_{1}}b_{\theta_{2}}}$). Observe that (1)
if $a_{\theta_{1}}a_{\theta_{2}}<1$ (resp., $b_{\theta_{1}}b_{\theta_{2}}$\textgreater 1)
then both eigenvalues are negative, which implies that the pure state
$\boldsymbol{a}$ (resp., $\boldsymbol{b}$) is asymptotically stable,
and (2) if $a_{\theta_{1}}a_{\theta_{2}}>1$ (resp., $b_{\theta_{1}}b_{\theta_{2}}$\textgreater 1)
then one of the eigenvalues is positive, which implies that this state
is unstable (see, e.g., \citealp[Theorems 1 and 2 in Section 2.9]{perko2013differential}).
\end{proof}
 Next we show that if any interior initial state converges to one of the pure equilibria, then this equilibrium is asymptotically stable.
\begin{lem}
\label{cla-convergence-to-pure-iff-asymp} Assume that $\textbf{p}(0)\neq(0,0)$
and $\lim_{t\rightarrow\infty}\mathbf{p}\left(t\right)=(0,0)$; then
$(0,0)$ is asymptotically stable. The same result holds when $(1,1)$ replaces $(0,0).$ 
\end{lem}
\begin{proof}
If $\textbf{p}$ is below (resp., above) both curves, then by the
classification presented in the proof of Lemma  \ref{lem:conv-to-area-between-curves}
it must be that $\dot{p_{2}}>0$ (resp., $\dot{p_{1}}>0$), which implies
that convergence to (0,0) is possible only if the trajectory passes
through a state that is strictly between the two curves, and that
the closest intersection point of the two curves to the left of this
state is (0,0). By the classification presented in the proof of Lemma
\ref{lem:convergence-to-which-neighbour} it must be that the curve of
$w_{2}(p_{1})$ is strictly below the curve of $w_{1}(p_{2})$ in
a right neighborhood of (0,0), which implies that (0,0) is asymptotically
stable because any sufficiently close initial state would converge
to (0,0). 
\end{proof}
We now use Proposition \ref{prop:pure-equilbiria-satbility} and Lemma \ref{cla-convergence-to-pure-iff-asymp} to prove Theorem \ref{thm4}.
\begin{enumerate}
\item Proposition \ref{prop:pure-equilbiria-satbility} implies that both
pure stationary states are unstable. Part (5) of Proposition \ref{prop-w-stable} and Lemma \ref{cla-convergence-to-pure-iff-asymp} imply that from any interior state, the population converges
to an interior stationary state. 
\item Proposition \ref{prop:pure-equilbiria-satbility} implies that at least one of the pure equilibria is asymptotically stable, which implies
that some interior initial states converge to a pure equilibrium.

\end{enumerate}

\begin{adjustwidth}{-5pt}{-11pt}
\begin{singlespace}
\bibliographystyle{chicago}
\bibliography{mybibdata}
\end{singlespace}
\end{adjustwidth}

\newpage
\pagenumbering{arabic}

\section{Online Appendix\label{sec:appendix}}\label{online-appendix}

 \subsection{Standard Representation of Coordination Games}\label{subsec-general-coord}
\paragraph{Symmetric Coordination Games}
We first show that the one-parameter standard representation (left panel of Table \ref{tab-symmetric}) w.l.o.g. captures 
 any two-action symmetric coordination game (in line with \citeauthor*{harsanyi1988general}'s (\citeyear{harsanyi1988general}) Axiom 2 of best-response invariance). 

A two-action symmetric coordination game is characterized by two strict equilibria on the main diagonal, as represented by the right panel of Table \ref{tab-symmetric}. Sampling dynamics (defined in \ref{eq:dynamics-sym} and \ref{eq:action-sampling-dyanmics}) depend solely on the differences between the payoffs a player can achieve by choosing different actions. (This property also holds for best-response and logit dynamics, implying that the sets of Nash equilibria, quantal response equilibria, and evolutionarily stable strategies depend only on these payoff differences.) These differences remain unchanged when a constant is subtracted from all of a player’s payoffs while holding the opponent's action fixed
(e.g., subtracting $u_{21}$ from all of Player 1's first-column payoffs).
Additionally, the sampling dynamics (and the other dynamics and solution concepts mentioned above) are invariant to dividing a player’s payoffs by a positive constant (which preserves vN--M utility). The right panel of Table \ref{tab-symmetric} is reduced to the left panel by the following steps (which do not affect the sampling dynamics): 
\begin{enumerate}
\item Subtracting\textcolor{blue}{{}
$u_{21}$} from Player 1's payoffs in her first column. 
\item Subtracting\textcolor{blue}{{}
$u_{12}$} from Player 1's payoffs in her second column.
\item Dividing
Player 1's payoffs by ${\color{blue}u_{22}-u_{12}}$
\end{enumerate}

\paragraph{General Coordination Games}\label{subsec-general-coord}
Next, we show that two-parameter standard representation (Table \ref{tab:standard-coordination-game-gen}, with $u_1\geq1$) can capture 
any two-action game that admits two strict Nash
equilibria. By relabeling Player 1's actions, we can assume w.l.o.g.
that these two pure equilibria are $\left(a_{1},a_{2}\right)$ and
$\left(b_{1},b_{2}\right).$ If the two pure equilibria are instead
$\left(a_{1},b_{2}\right)$ and $\left(b_{1},a_{2}\right)$, we  switch the Player 1's action labels: $a_{1}$$\leftrightarrow$$b_{1}).$
This implies that the left panel of Table \ref{tab:normalization-of-coordiantion-game}
shows a parametric representation of all two-action coordination games.

\begin{table}[b]
\caption{\label{tab:normalization-of-coordiantion-game}Normalization of General
Two-Action Coordination Games}

\centering{}%
\begin{tabular}{|c|c|c|c|c|c|c|}
\multicolumn{3}{c}{Original Representation} & \multicolumn{1}{c}{} & \multicolumn{3}{c}{Standard Representation}\tabularnewline
\cline{1-3} \cline{2-3} \cline{3-3} \cline{5-7} \cline{6-7} \cline{7-7} 
 & \textcolor{red}{\emph{$a_{2}$}} & \textcolor{red}{\emph{$b_{2}$}} & \multirow{3}{*}{$\begin{array}{c}
\\
\\
\Rightarrow\\
\\
\\
\end{array}$} &  & \textcolor{red}{\emph{$a_{2}$}} & \textcolor{red}{\emph{$b_{2}$}}\tabularnewline
\cline{1-3} \cline{2-3} \cline{3-3} \cline{5-7} \cline{6-7} \cline{7-7} 
\textcolor{blue}{\emph{$~~a_{1}~~$}} & \emph{\Large{}$\,\,{}_{\underset{\,}{{\color{blue}u_{11}}}}{\color{red}\,^{\overset{\,}{v_{11}}}}$~~} & \emph{\Large{}$\,\,{}_{\underset{\,}{{\color{blue}u_{12}}}}{\color{red}\,^{\overset{\,}{v_{12}}}}$~~} &  & \textcolor{blue}{\emph{$~~a_{1}~~$}} & \emph{\Large{}$_{\color{blue}u_1=}\frac{{\color{blue}u_{11}}-{\color{blue}u_{21}}}{{\color{blue}u_{22}-}{\color{blue}u_{1\underset{\,}{2}}}}{\color{red}\,^{\overset{}{u_2=\frac{v_{11}-v_{12}}{v_{22}-v_{21}}}}}$} & \emph{\Large{}$\,\,\,\,{}_{\underset{\,}{{\color{blue}0}}}{\color{red}\,\,\,\,^{\overset{\,}{0}}}$~~~~}\tabularnewline
\cline{1-3} \cline{2-3} \cline{3-3} \cline{5-7} \cline{6-7} \cline{7-7} 
\textcolor{blue}{\emph{$~~b_{1}~~~$}} & \emph{\Large{}$\,\,{}_{\underset{\,}{{\color{blue}u_{21}}}}{\color{red}\,^{\overset{\,}{v_{21}}}}$~~} & \emph{\Large{}$\,\,{}_{\underset{\,}{{\color{blue}u_{22}}}}{\color{red}\,^{\overset{\,}{v_{22}}}}$~~} &  & \textcolor{blue}{\emph{$~~b_{1}~~~$}} & \emph{\Large{}$\,\,{}_{\underset{\,}{{\color{blue}0}}}{\color{red}\,\,\,\,\,\,\,\,\,^{\overset{\,}{0}}}$~~} & \emph{\Large{}$\,\,{}_{\underset{\,}{{\color{blue}1}}}\,\,\,{\color{red}\,^{\overset{\,}{1}}}$~~}\tabularnewline
\cline{1-3} \cline{2-3} \cline{3-3} \cline{5-7} \cline{6-7} \cline{7-7} 
\multicolumn{7}{c}{\textcolor{blue}{$\underset{}{{\color{blue}u_{11}>u_{21}},u_{22}>u_{12}}$,}\textcolor{red}{{}
$\underset{}{v_{11}>v_{12}},\underset{}{v_{22}>v_{21}}$}}\tabularnewline
\end{tabular}
\end{table}
As mentioned above, sampling dynamics (as defined in \ref{eq:action-sampling-dyanmics-g}) depend only on the differences between the payoffs
a player can get by playing different actions. These differences are invariant to (1) subtracting a constant from all the payoffs of a player while fixing the opponent's action, and (2) dividing
player's payoffs by a positive constant (which preserves the
vN--M utility). The left panel of Table \ref{tab:normalization-of-coordiantion-game}
is reduced to the right panel by the following steps (none of which affect the sampling dynamics):
\begin{enumerate}
\item Three changes to Player 1's payoffs: (I) subtracting\textcolor{blue}{{}
$u_{21}$} from Player 1's payoffs in her first column, (II) subtracting\textcolor{blue}{{}
$u_{12}$} from Player 1's payoffs in her second column, and (III) dividing
Player 1's payoffs by ${\color{blue}u_{22}-}u_{12}$; and 
\item Three changes to Player 2's payoffs: (I) subtracting \textcolor{red}{$v_{12}$}
from Player 1's payoffs in her first row, (II) subtracting \textcolor{blue}{${\color{red}v_{21}}$}
from Player 1's payoffs in her first row, and (III) dividing Player 1's
payoffs by ${\color{red}v_{22}-v_{21}}$.
\end{enumerate}
Observe that the assumption that ${\color{blue}u_{1}=}\frac{{\color{blue}u_{11}}-{\color{blue}u_{21}}}{{\color{blue}u_{22}-}{\color{blue}u_{12}}}\geq1$
in the standard representation of Table \ref{tab:standard-coordination-game-gen}
is w.l.o.g.. If $\frac{{\color{blue}u_{11}}-{\color{blue}u_{21}}}{{\color{blue}u_{22}-}{\color{blue}{\color{blue}u_{12}}}}<1$,
then we can multiply all of Player 1's payoffs by $\frac{{\color{blue}u_{22}-}{\color{blue}{\color{blue}u_{12}}}}{{\color{blue}u_{11}}-{\color{blue}u_{21}}}$
and all of Player 2's payoffs by ${\color{red}\frac{v_{22}-v_{21}}{v_{11}-v_{12}}}$,
relabel the actions $a_{i}\leftrightarrow b_{i}$ for both players,
and obtain a standard representation in which $u_{1}\geq1$ as in Table \ref{tab:standard-coordination-game-gen}.
\begin{example}
\label{exa:motivating-hawk--dove} Consider a \emph{hawk--dove} (aka chicken)
game, 
which can be interpreted as bargaining over the price of an asset
(e.g., a house) between a buyer and a seller. Each player can either
insist on a more favorable price (“hawk”) or agree to a less favorable
price in order to close the deal (“dove”). The left panel
of Table \ref{tab:normalization-of-hawk--dove-games} shows the 
payoffs of a hawk--dove game. Two doves agree on an equally favorable
price 
(which gives both players a relatively high payoff normalized to one). A hawk obtains a favorable price when matched with a
dove (which increases the payoff to the hawk by $g\in(0,1)$, while reducing the dove's payoff by $l\in(0,1)$), but faces a high probability of bargaining failure when matched with
another hawk (which yields a low payoff of zero to both hawks).

Observe that a hawk--dove game can be transformed to our standard representation
of a coordination game (the right panel of Table \ref{tab:normalization-of-hawk--dove-games})
as follows: (1) relabel the actions of Player 1 such that $a_{1}=d_{1}$
and $b_{1}=h_{1}$ (while keeping the actions of Player 2 as $a_{2}=h_{2}$ and $b_{2}=d_{2}$), (2)
subtract a payoff of 1 from Player 1's payoffs in her second column
and from Player 2's payoffs in her first column, and (3) divide all
the payoffs of Player 1 by $g$, and all of the payoffs of Player 2 by $1-l$.
Observe that the induced standard representation is 
antisymmetric, i.e.,  $u_{1}=\frac{{\color{blue}1-l}}{g}=\frac{1}{u_{2}}$. 

\begin{table}
\caption{\label{tab:normalization-of-hawk--dove-games}Normalization of Hawk--Dove
Games ($g,l\in\left(0,1\right)$)}

\centering{}%
\begin{tabular}{|c|c|c|c|c|c|c|}
\multicolumn{3}{c}{Original Representation} & \multicolumn{1}{c}{} & \multicolumn{3}{c}{Standard Representation}\tabularnewline
\cline{1-3} \cline{2-3} \cline{3-3} \cline{5-7} \cline{6-7} \cline{7-7} 
 & \textcolor{red}{\emph{$h_{2}$}} & \textcolor{red}{\emph{$d_{2}$}} & \multirow{3}{*}{$\begin{array}{c}
\\
\\
\Rightarrow\\
\\
\\
\end{array}$} &  & \textcolor{red}{\emph{$a_{2}=h_{2}$}} & \textcolor{red}{\emph{$b_{2}=d_{2}$}}\tabularnewline
\cline{1-3} \cline{2-3} \cline{3-3} \cline{5-7} \cline{6-7} \cline{7-7} 
\textcolor{blue}{\emph{$~h_{1}$}} & \emph{\Large{}$\,\,{}_{\underset{\,}{{\color{blue}0}}}{\color{red}\,^{\overset{\,}{0}}}$~~} & $_{\underset{\,}{{\color{blue}1+g}}}\,\,\,{\color{red}\,^{\overset{\,}{1-l}}}$ &  & \textcolor{blue}{\emph{$~~~a_{1}=d_{1}~$}} & $\,\,{}_{\underset{\,}{{\color{blue}\frac{{\color{blue}1-l}}{g}}}}\,\,\,{\color{red}\,^{\overset{\,}{\frac{g}{1-l}}}}$ & \emph{\Large{}$\,{}_{\underset{\,}{{\color{blue}0}}}{\color{red}\,^{\overset{\,}{0}}}$~}\tabularnewline
\cline{1-3} \cline{2-3} \cline{3-3} \cline{5-7} \cline{6-7} \cline{7-7} 
\textcolor{blue}{\emph{$~d_{1}$}} & $\,\,{}_{\underset{\,}{{\color{blue}1-l}}}\,\,\,{\color{red}\,^{\overset{\,}{1+g}}}$ & \emph{\Large{}$_{\underset{\,}{{\color{blue}1}}}\,\,\,{\color{red}\,^{\overset{\,}{1}}}$} &  & \textcolor{blue}{\emph{$~~~b_{1}=h_{1}~~$}} & \emph{\Large{}$\,\,{}_{\underset{\,}{{\color{blue}0}}}{\color{red}\,\,\,\,\,\,\,^{\overset{\,}{0}}}$~~} & \emph{\Large{}$\,\,{}_{\underset{\,}{{\color{blue}1}}}{\color{red}\,^{\overset{\,}{1}}}$~~}\tabularnewline
\hline 
\end{tabular}
\end{table}
\end{example}

\subsection{Rephrasing Our Results in Terms of $q$-Dominance}\label {subsec:stability-pure-q-dominance}
In this appendix, we rephrase our results in terms of $q$-dominance (\cite{morris1995p} of the equilibria, rather than payoffs $(u_1,u_2)$. This serves two purposes: (1) to facilitate comparison with the literature, which often presents results in terms of $q$-dominance (e.g., \citealp{oyama2015sampling}), and (2) to provide a clearer interpretation of the results that is invariant to payoff normalization. Specifically, normalizing a game to fit the standard two-parameter representation (see Appendix \ref{subsec-general-coord}), alters the payoffs, yet preserves the $q$-dominance of the equilibria.

Fix $q\in\left[0,1\right]$. 
We say that action
$a_{i}$ (resp., $b_{i}$)
is \emph{$q$-dominant} (\citealp{morris1995p})
for player $i$ if it is a strict best response to any opponent's mixed
action that assigns a mass of at least $q$ to the 
counterpart action
$a_{j}$ (resp., $b_j$). Notice that both actions are 1-dominant (due
to being part of a strict equilibrium). 
Additionally, it can be noted that as $q$ decreases, the $q$-dominance condition becomes more stringent. In other words, if an action is $q$-dominant, it also satisfies $r$-dominance for any $r$ between $q$ and 1. 
Lastly, it can be observed that action $a_i$ (resp., $b_i$)
is $q$-dominant iff $q>\frac{1}{1+u_i}$ (resp., $q>\frac{u_i}{1+u_i}$).

Observe that $q$-dominance depends only on the differences between
the payoffs a player can get by playing the different actions. This
implies that $q$-dominance is invariant to the payoff transformations detailed in Appendix \ref{subsec-general-coord}. Thus, an action is $q$-dominant
in the standard representation (left panel of Table \ref{tab:normalization-of-coordiantion-game})
iff it is $q$-dominant in the original representation (right panel of Table \ref{tab:normalization-of-coordiantion-game}).
By contrast, payoff dominance is not invariant to these two transformations. Specifically, adding a constant to the two payoffs of 
Player 1 (Player 2) in the same column (row) might change a Pareto-dominated equilibrium into a Pareto-dominant equilibrium.

The following definition will be useful for our characterization. 
\begin{defn}
Action profile $\textbf{a}$ (resp., $\textbf{b}$) is tightly $\textbf{q}$-dominant if each action $a_i$ (resp., $b_i$) is $r_i$-dominant iff $r_i>q_i$.
\end{defn}

Observe that:
\begin{enumerate}
\item $\textbf{a}$ is  tightly $(q_1,q_2)$-dominant iff $\textbf{b}$ is  tightly $(1-q_1,1-q_2)$-dominant. 
    \item If the payoffs of the coordination game (in its standard representation) are $(u_1,u_2)$, then equilibrium $\textbf{a}$ is  tightly ($\frac{1}{1+u_1},\frac{1}{1+u_2})$-dominant.
    \item if equilibrium $\textbf{a}$ is tightly $(q_1,q_2)$-dominant, then the standard representation of the payoff matrix is $u_i=\frac{1-q_i}{q_i}.$
\end{enumerate} 

We redefine an environment as a pair $(\textbf{q},
\boldsymbol{\theta})$, 
where $\textbf{q}$ represents the level of tight dominance of equilibrium $\textbf{a}$ and $\boldsymbol{\theta})$ is the sample-size distribution profile. W.l.o.g. we assume $q_1\geq\frac{1}{2}.$
The above observations imply the following rephrasing of our results.
\theoremstyle{plain}
\newtheorem*{theorem2hat}{Theorem $\hat{2}$}
\begin{theorem2hat}[Rephrasing of Theorem \ref{thm:locally-stable_interior_if_pure_is_unstable}]
If for each population $i$, $1<\max(\emph{supp}(\theta_i))< \frac{1}{q_i}.$ Then, there exists  a proportion $\alpha_i$ such that significantly increasing the sample sizes of $\alpha_i$ of the agents in each population $i$ induces an environment with an asymptotically stable interior state.
\end{theorem2hat}

%\begin{thm}[Rephrasing of Theorem \ref{thm2'}]
%If for each population $i$, $1<\max(\text{supp}(\theta_i))< \frac{1}{q_i}.$ Then, there exists  a proportion $\alpha_i$ such that significantly increasing the sample sizes of $\alpha_i$ of the agents in each population $i$ induces an environment with an asymptotically stable interior state.
%\end{thm}

\theoremstyle{plain}
\newtheorem*{theorem3hat}{Theorem $\hat{3}$}
\begin{theorem3hat}[Rephrasing of Theorem \ref{thm3}]
For any sample size distribution profile, if $q_1$ is sufficiently small and $q_2$ is sufficiently large, then significantly increasing the sample size of half of the agents in each population induces an asymptotically stable interior state with a miscoordination probability of at least 50\%.
\end{theorem3hat}

%\begin{thm}[Rephrasing of Theorem \ref{thm3}]
%For any sample size distribution profile, if $q_1$ is sufficiently small and $q_2$ is sufficiently large, then significantly increasing the sample size of half of the agents in each population induces an asymptotically stable interior state with a miscoordination probability of at least 50\%. 
%\end{thm}

\theoremstyle{plain}
\newtheorem*{theorem4hat}{Theorem $\hat{4}$}
\begin{theorem4hat}[Rephrasing of Theorem \ref{thm4}]
\,\\
\vspace{-20px}
\begin{enumerate}
\item \textbf{Global convergence to miscoordination}: Assume that \vspace{-10px}
\[
\theta_{1}\left(1\right)\cdot\mathbb{E}_{<\frac{1}{1-q_2}}\left(\theta_{2}\right)>1\,\,\textrm{\,\,and}\,\,\,\,\theta_{2}\left(1\right)\cdot\mathbb{E}_{
< \frac{1}{q_1}}\left(\theta_{1}\right)>1.
\]
If $\boldsymbol{p}(0)\notin\left\{ (0,0),(1,1)\right\} $, then $\lim_{t\rightarrow\infty}\mathbf{p}\left(t\right)\notin\left\{ (0,0),(1,1)\right\} $. 
\item \textbf{Local convergence to coordination}: Assume that 
\[
\theta_{1}\left(1\right)\cdot\mathbb{E}_{<\frac{1}{1-q_2}}\left(\theta_{2}\right)<1\,\,\textrm{\,\,or}\,\,\,\,\theta_{2}\left(1\right)\cdot\mathbb{E}_{
< \frac{1}{q_1}}\left(\theta_{1}\right)<1.
\]
Then at least one of the pure equilibria is asymptotically stable.
\end{enumerate}
\end{theorem4hat}

\theoremstyle{plain}
\newtheorem*{prop4hat}{Proposition $\hat{4}$}
\begin{prop4hat}
Assume that pure equilibrium $\textbf{a}$ is tightly $(q_1,q_2)$-dominant. Then\footnote{The slight difference in using strict inequalities for $\textbf{a}$ and weak inequalities for $\textbf{b}$ is due to our $a_i$-favorable tie-breaking rule.} 
 
\begin{enumerate}
    \item $\mathbb{E}_{<\frac{1}{1-q_1}}\left(\theta_{1}\right)\cdot\mathbb{E}_{<\frac{1}{1-q_2}}\left(\theta_{2}\right)>1\,\Rightarrow$ $\left(a_{1},a_{2}\right)$ is unstable;
    \item $\mathbb{E}_{<\frac{1}{1-q_1}}\left(\theta_{1}\right)\cdot\mathbb{E}_{<\frac{1}{1-q_2}}\left(\theta_{2}\right)<1\,\Rightarrow$ $\left(a_{1},a_{2}\right)$ is asymptotically stable;
    \item $\mathbb{E}_{\leq\frac{1}{q_1}}\left(\theta_{1}\right)\cdot\mathbb{E}_{\leq\frac{1}{q_2}}\left(\theta_{2}\right)>1\,\Rightarrow$ $\left(b_{1},b_{2}\right)$ is unstable; and
    \item $\mathbb{E}_{\leq\frac{1}{q_1}}\left(\theta_{1}\right)\cdot\mathbb{E}_{\leq\frac{1}{q_2}}\left(\theta_{2}\right)<1\,\Rightarrow$ $\left(b_{1},b_{2}\right)$ is asymptotically stable.
\end{enumerate}
\end{prop4hat}

\subsection{Coordination Games with More Than Two Actions}\label{sec-multiple-actions}

\subsubsection{Extended model}
We redefine the underlying game as a two-player
coordination game with $M\geq2$ actions. The action sets $A_{i}=\left(a_{i}^{1},...,a_{i}^{M}\right)$ are finite, with positive payoffs on the main diagonal and zero off-diagonal payoffs:  (1) $u_{i}^{m}\equiv u_{i}\left(a_{1}^{m},a_{2}^{m}\right)>0$
for each $m$, and (2) $u_{i}\left(a_{1}^{m},a_{2}^{n}\right)=0$
for each $m\neq n$. The payoff matrix is shown in Table \ref{tab:Payoff-Matrix-for-coordination-games-M-actions}. To insure results hold for all tie-breaking rules, we take the following mild genericity assumption: there do not exist two pure equilibria that yield the same payoff profile
(i.e., if $u_i^m={u_i^m}'$ then $u_j^m\neq {u_j^m}'$).

This class of coordination games is significant as it captures common
situations in which players must collectively agree on the terms of
a joint venture that could yield positive benefits for both parties
involved. Each action profile on the main diagonal represents potential
mutually agreed-upon terms, while off-diagonal action profiles represent
disagreement, resulting in the failure of the joint venture and, consequently,
a low payoff (normalized to zero) for both players. These games  are called pure coordination games or contracting games
(see, e.g., \citealp{young1998conventional, hwang2017payoff}).
\begin{center}
\begin{table}
\begin{centering}
\caption{\label{tab:Payoff-Matrix-for-coordination-games-M-actions}Payoff
Matrix for Coordination Games with $M\protect\geq2$ Actions}
\begin{tabular}{|c|c|c|c|c|c|}
\hline 
 & \textcolor{red}{$a_{2}^{1}$} & \textcolor{red}{...} & \textcolor{red}{$a_{2}^{m}$} & \textcolor{red}{...} & \textcolor{red}{$a_{2}^{M}$}\tabularnewline
\hline 
\hline 
\textcolor{blue}{$a_{1}^{1}$} & ${\color{blue}u_{1}^{1}},\,{\color{red}u_{2}^{1}}$ & \textcolor{blue}{0},\textcolor{red}{0} & \textcolor{blue}{0},\textcolor{red}{0} & \textcolor{blue}{0},\textcolor{red}{0} & \textcolor{blue}{0},\textcolor{red}{0}\tabularnewline
\hline 
\textcolor{blue}{...} & \textcolor{blue}{0},\textcolor{red}{0} & ... & \textcolor{blue}{0},\textcolor{red}{0} & \textcolor{blue}{0},\textcolor{red}{0} & \textcolor{blue}{0},\textcolor{red}{0}\tabularnewline
\hline 
\multicolumn{1}{|c|}{\textcolor{blue}{$a_{1}^{m}$}} & \textcolor{blue}{0},\textcolor{red}{0} & \textcolor{blue}{0},\textcolor{red}{0} & ${\color{blue}u_{1}^{m}},\,{\color{red}u_{2}^{m}}$ & \textcolor{blue}{0},\textcolor{red}{0} & \textcolor{blue}{0},\textcolor{red}{0}\tabularnewline
\hline 
\textcolor{blue}{...} & \textcolor{blue}{0},\textcolor{red}{0} & \textcolor{blue}{0},\textcolor{red}{0} & \textcolor{blue}{0},\textcolor{red}{0} & ... & \textcolor{blue}{0},\textcolor{red}{0}\tabularnewline
\hline 
\textcolor{blue}{$a_{1}^{M}$} & \textcolor{blue}{0},\textcolor{red}{0} & \textcolor{blue}{0},\textcolor{red}{0} & \textcolor{blue}{0},\textcolor{red}{0} & \textcolor{blue}{0},\textcolor{red}{0} & ${\color{blue}u_{1}^{M}},\,\,{\color{red}u_{2}^{M}}$\tabularnewline
\hline 
\end{tabular}
\par\end{centering}
\end{table}
\par\end{center}
\vspace{-5px}
We redefine an environment as  a tuple $\left(G,\boldsymbol{\theta}\right)$,
where $G$ is a two-player coordination game with $M\geq2$ actions, and $\boldsymbol{\theta}$
is the profile of sample size distributions. A state of population
$i$ is a distribution $p_{i}\in\Delta\left(A_{i}\right)$ over the
actions of player $i$. As in the baseline model, each new agent
with sample size $k$ samples $k$ random actions of her opponent and
plays the action that maximizes her payoff against the
sample. The results of this section hold under any tie-breaking rule.
All other aspects of the model remain unchanged.

A pure equilibrium $\boldsymbol{a}^{m}=\left(a_{1}^{m},a_{2}^{m}\right)$
is \emph{Pareto efficient} if for each $n$,
$u_{1}^{m}<u_{1}^{n}$
implies that $u_{2}^{m}>u_{2}^{n}$.
Let $\bar{u}_{i}$ denote the
highest feasible payoff of player $i$, i.e., $\bar{u_{i}}=\max_{m\leq M}\left(u_{i}^{m}\right)$.
Let $\bar{m}_{i}$ be the index of an action that induces payoff
$\bar{u_{i}}$, i.e., $u_{i}^{\bar{m}_{i}}=\bar{u_{i}}$.

\subsubsection{Generalized Result}

In what follows, we generalize Theorem \ref{thm:global-mixed} to coordination games with
$M\geq2$ actions. Specifically, we show that similar to the case
of two actions, the stability of each pure equilibrium depends on
whether the product of the share of agents with sample size one and
the truncated expectation of the sample size is larger or smaller than
one. Formally,\footnote{The fact that the truncated expectation has strict inequality in part
(1) and weak inequality in part (2) allows these conditions to be
valid under any tie-breaking rule.}
\begin{prop}[Generalization of Theorem \ref{thm:global-mixed}]~\\
\label{thm-1-generelized}
\vspace{-20px}
\begin{enumerate}
\item Assume that for each pure equilibrium $\boldsymbol{a}^{m},$
\vspace{-10px}
\[
\theta_{1}\left(1\right)\cdot\mathbb{E}_{<\frac{\bar{u}_{2}}{u_{2}^{m}}+1}\left(\theta_{2}\right)>1\,\,\textrm{or}\,\,\theta_{2}\left(1\right)\cdot\mathbb{E}_{<\frac{\bar{u}_{1}}{u_{1}^{m}}+1}\left(\theta_{1}\right)>1.
\]

\vspace{-15px} Then all pure equilibria are unstable.
\vspace{-10px}
\item Assume that there exists a Pareto-efficient pure equilibrium $\boldsymbol{a}^{m}$
that satisfies, 
\vspace{-10px}
\[
\theta_{1}\left(1\right)\cdot\mathbb{E}_{\leq\frac{\bar{u}_{2}}{u_{2}^{m}}+1}\left(\theta_{2}\right)<1\,\,\textrm{and}\,\,\theta_{2}\left(1\right)\cdot\mathbb{E}_{\leq\frac{\bar{u}_{1}}{u_{1}^{m}}+1}\left(\theta_{1}\right)<1.
\]

\vspace{-18px} Then equilibrium $\boldsymbol{a}^{m}$ is asymptotically stable.
\qedhere \vspace{10px} 
\end{enumerate}
\end{prop}
\begin{proof}[Sketch of Proof]
 See Appendix \ref{apx-proof-thm-4} for a formal proof.
\begin{enumerate}
\item Assume that $\theta_{1}\left(1\right)\cdot\mathbb{E}_{<\frac{\bar{u}_{2}}{u_{2}^{m}}+1}\left(\theta_{2}\right)>1$
(resp., $\theta_{2}\left(1\right)\cdot\mathbb{E}_{<\frac{\bar{u}_{1}}{u_{1}^{m}}+1}\left(\theta_{1}\right)>1$).
Observe that in any initial state in which in each population $i$
almost all agents play $a_{i}^{m}$, while a few agents play $a_{i}^{\bar{m}_{2}}$
(resp., $a_{i}^{\bar{m}_{1}}$), the product of the shares of agents
who play $a_{i}^{\bar{m}_{2}}$ (resp., $a_{i}^{\bar{m}_{1}}$) in
each population $i$ would increase by analogous arguments to those in
Proposition \ref{prop:pure-equilbiria-satbility}. This implies that equilibrium $\boldsymbol{a}^{m}$
is unstable. If condition (1) holds for all pure equilibria, then
all of those equilibria are unstable.
\item Consider any initial state in which almost all agents play $a_{i}^{m}$. The fact that $\boldsymbol{a}^{m}$ is Pareto efficient
implies that $u_{\hat{i}}^{m'}< u_{i}^{m}$ for some $i\in\left\{ 1,2\right\} $.
By analogous arguments to those in Proposition \ref{prop:pure-equilbiria-satbility}, the small share of agents who
play $a_{i}^{m'}$ would decrease if $\theta_{i}\left(1\right)\mathbb{E}_{\leq\frac{u_{j}^{m'}}{u_{j}^{m}}+1}\left(\theta_{j}\right)<1$. This inequality holds because  $\mathbb{E}_{\leq\frac{u_{j}^{m'}}{u_{j}^{m}}+1}\left(\theta_{j}\right)\leq\mathbb{E}_{\leq\frac{\bar{u}_{j}}{u_{j}^{m}}+1}\left(\theta_{2}\right)$
and $\theta_{i}\left(1\right)\mathbb{E}_{\leq\frac{\bar{u}_{j}}{u_{j}^{m}}+1}\left(\theta_{2}\right)<1$.
This implies that 
$\boldsymbol{a}^{m}$ is asymptotically stable.\qedhere
\end{enumerate}
\end{proof}

Proposition \ref{thm-1-generelized} immediately implies that if all agents have the same sample size  size, then all pure states are asymptotically stable. \begin{cor}[Adaptation of Theorem \ref{thm:global-stability-pure}]
Assume that $\theta_{i}\equiv k_{i}>1$ for each $i\in\left\{ 1,2\right\} $. Then, all the Pareto-efficient pure equilibria are asymptotically stable.
\end{cor}

Thus, heterogeneity in sample size is important for the stability of miscoordination also
in this extended setup.
 Finally, observe that it is straightforward to adapt Theorem 3 to
the setup with $M\geq2$ actions, by assuming that (1) the payoffs
of two of the pure equilibria satisfy the conditions of Theorem 3,
and (2) these two equilibria Pareto dominate the remaining pure equilibria.
Formally,
\begin{cor}[Generalization of Theorem \ref{thm:locally-stableinterior_ASD}]\label{cor-locally-stable-interior-M}
If $\frac{u^1_1}{u^2_1}$ and $\frac{u^1_2}{u^2_2}$ are not too close to one and the $u^m_i$-s are sufficiently small for each $m>2$, then replacing some of the agents in each population by agents with sufficiently large sample sizes, can induce an asymptotically stable interior state $\boldsymbol{\hat{p}}$.\\ 
\end{cor}
\noindent The proof (which is essentially the same as the proof of Theorem 3) is omitted for brevity.

\subsubsection{Formal Proof of Proposition \ref{thm-1-generelized}}\label{apx-proof-thm-4} 
\paragraph {Part 1.} 
We have to show that pure equilibrium $\boldsymbol{a}^{m}$ is unstable
if either $\theta_{1}\left(1\right)\cdot\mathbb{E}_{<\frac{\bar{u}_{2}}{u_{2}^{m}}+1}\left(\theta_{2}\right)>1\,\,\textrm{or}\,\,\theta_{2}\left(1\right)\cdot\mathbb{E}_{<\frac{\bar{u}_{1}}{u_{1}^{m}}+1}\left(\theta_{2}\right)>1.$
Assume that $\theta_{1}\left(1\right)\cdot\mathbb{E}_{<\frac{\bar{u}_{2}}{u_{2}^{m}}+1}\left(\theta_{2}\right)>1$
(resp., $\theta_{2}\left(1\right)\cdot\mathbb{E}_{<\frac{\bar{u}_{1}}{u_{1}^{m}}+1}\left(\theta_{2}\right)>1$).
Consider a slightly perturbed state near $\boldsymbol{a}^{m}$, where
in each population $i$ a small share $\epsilon_{i}<<1$ of the agents
play action $a_{i}^{\bar{m}_{2}}$ (resp., $a_{i}^{\bar{m}_{1}}$),
while all the other agents play action $a_{i}^{m}$. Observe that: (1)
a new agent in population 1 (resp., 2) with sample size 1 who observes
the rare action $a_{2}^{\bar{m}_{2}}$ (resp., $a_{1}^{\bar{m}_{1}}$)
plays action $a_{1}^{\bar{m}_{2}}$ (resp., $a_{2}^{\bar{m}_{1}}$),
and (2) a new agent in population 2 (resp., 1) with sample size $k$
who observes the rare action $a_{1}^{\bar{m}_{2}}$ (resp., $a_{2}^{\bar{m}_{1}}$)
once in her sample plays $a_{2}^{\bar{m}_{2}}$ (resp., $a_{1}^{\bar{m}_{1}}$)
if $\left(k-1\right)u_{2}^{m}<\bar{u}_{2}\Leftrightarrow k<\frac{\bar{u}_{2}}{u_{2}^{m}}+1$
(resp., $k<\frac{\bar{u}_{1}}{u_{1}^{m}}+1$). This gives the following
lower bound for the change in the share of agents who play the rare
action (neglecting terms that are of order $O\left(\epsilon_{i}^{2}\right)$): 
\[
\dot{\epsilon}_{1}\geq\theta_{1}\left(1\right)\epsilon_{2}-\epsilon_{1}\,\,\textrm{and}\,\,\dot{\epsilon}_{2}\geq\mathbb{E}_{<\frac{\bar{u}_{2}}{u_{2}^{m}}+1}\left(\theta_{2}\right)\epsilon_{1}-\epsilon_{2}
\]
 \vspace{-20px}
\[
\left(\textrm{resp.,\,\,}\dot{\epsilon}_{1}\geq\mathbb{E}_{<\frac{\bar{u}_{1}}{u_{1}^{m}}+1}\left(\theta_{1}\right)\epsilon_{1}-\epsilon_{2}\,\,\textrm{and}\,\,\dot{\epsilon}_{2}\geq\theta_{2}\left(1\right)\epsilon_{1}-\epsilon_{2}\right).
\vspace{-10px}
\] 
Observe that the Jacobian of the above system of equations is given
by $J=\left(\begin{array}{cc}
-1 & a_{1}\\
a_{2} & -1
\end{array}\right)$ with $a_{1}\geq\theta_{1}\left(1\right)$ and $a_{2}\geq\mathbb{E}_{<\frac{\bar{u}_{2}}{u_{2}^{m}}+1}\left(\theta_{2}\right)$
(resp., $a_{2}\geq\theta_{2}\left(1\right)$ and $a_{1}\geq\mathbb{E}_{<\frac{\bar{u}_{1}}{u_{1}^{m}}+1}\left(\theta_{1}\right)$).
The larger eigenvalue is given by $-1+\sqrt{a_{1}a_{2}}$, which is
larger than $-1+\sqrt{\theta_{1}\left(1\right)\mathbb{E}_{<\frac{\bar{u}_{2}}{u_{2}^{m}}+1}\left(\theta_{2}\right)}>-1+1=0$
(resp., $-1+\sqrt{\theta_{2}\left(1\right)\mathbb{E}_{<\frac{\bar{u}_{1}}{u_{1}^{m}}+1}\left(\theta_{2}\right)}>-1+1=0$). The fact that this eigenvalue is positive implies that the state is
unstable (see, e.g., \citealp[Section 2.9]{perko2013differential}).

\paragraph{Part 2:}
Any perturbed state near $\boldsymbol{a}^{m}$ can be represented
as a vector $\left(\epsilon_{1}^{m'},\epsilon_{2}^{m'}\right)_{m'\neq m}$
with $2\times\left(M-1\right)$ components, where $\epsilon_{i}^{m'}<<1$
represents the share of agents in population $i$ who play action
$a_{i}^{m'}$ (where the remaining share $1-\sum_{m'\neq m}\epsilon_{i}^{m'}$
of agents in population $i$ play action $a_{i}^{m}$). Observe
that a new agent in population $i$ with sample size $k$ who observes
the rare action $a_{j}^{m'}$ once in her sample (and all other observed
actions are $a_{j}^{m})$ plays $a_{j}^{m'}$ only if $\left(k-1\right)u_{i}^{m}\leq u_{i}^{m'}\Leftrightarrow k\leq\frac{u_{i}^{m'}}{u_{i}^{m}}+1$.
This give the following upper bound for the share of agents who play
the rare action (neglecting terms that are of order $O\left(\epsilon_{i}^{m'}\right)^{2}$):

\[
\dot{\epsilon}_{i}^{m'}\leq\mathbb{E}_{\leq\frac{u_{i}^{m'}}{u_{i}^{m}}+1}\left(\theta_{i}\right)\epsilon_{j}^{m'}-\epsilon_{i}^{m'}.
\]
 Observe that the Jacobian of the above system of ($2\times\left(M-1\right)$)
equations is given by 

\[
J=\left(\begin{array}{cccccccc}
-1 & a_{1}^{1} &  & 0 & 0 &  & 0 & 0\\
a_{2}^{1} & -1 &  & 0 & 0 &  & 0 & 0\\
 &  &  & \vdots\\
0 & 0 & ... & -1 & a_{1}^{m'} & ... & 0 & 0\\
0 & 0 &  & a_{2}^{m'} & -1 &  & 0 & 0\\
 &  &  & \vdots\\
0 & 0 &  & 0 & 0 &  & -1 & a_{1}^{M}\\
0 & 0 &  & 0 & 0 &  & a_{2}^{M} & -1
\end{array}\right),
\]
 where $a_{1}^{m'}\leq\mathbb{E}_{\leq\frac{u_{i}^{m'}}{u_{i}^{m}}+1}\left(\theta_{i}\right).$
The Jacobian has the following $2\times\left(M-1\right)$ eigenvalues:
$\left(-1\pm a_{1}^{m'}\cdot a_{2}^{m'}\right)_{m'\neq m}$. The fact
that  $\boldsymbol{a}^{m}$ is Pareto efficient (and
the mild assumption that there do not exist two pure equilibria that
yield the same payoff profile)
implies that $u_{i}^{m'}<u_{i}^{m}$
for $i\in\left\{ 1.2\right\} $. This implies that $\mathbb{E}_{\leq\frac{u_{i}^{m'}}{u_{i}^{m}}+1}=\theta_{i}\left(1\right)$.
Observe that $\mathbb{E}_{\leq\frac{u_{j}^{m'}}{u_{j}^{m}}+1}\left(\theta_{j}\right)\leq\mathbb{E}_{<\frac{\bar{u}_{j}}{u_{j}^{m}}+1}\left(\theta_{j}\right)$.
This, in turn, implies that the eigenvalue $-1\pm a_{1}^{m'}\cdot a_{2}^{m'}$
is bounded by 
\[
-1\pm a_{i}^{m'}\cdot a_{j}^{m'}\leq-1+\mathbb{E}_{\leq\frac{u_{i}^{m'}}{u_{i}^{m}}+1}\left(\theta_{i}\right)\mathbb{E}_{\leq\frac{u_{j}^{m'}}{u_{j}^{m}}+1}\left(\theta_{j}\right)\leq-1+\theta_{i}\left(1\right)\mathbb{E}_{\leq\frac{\bar{u}_{j}}{u_{j}^{m}}+1}\left(\theta_{2}\right)<-1+1=0,
\]
which implies that all the eigenvalues are positive. Thus, 
$\boldsymbol{a}^{m}$ is asymptotically stable.

\subsection{Coordination Games with More Than Two Players}\label{sec-multiple-players}

In this appendix, we extend our analysis to minimum effort coordination
games (\citealp{vanHuyck1990tacit}) with $N>2$ players. We analyze one-population dynamics of these games, as this aligns with the typical experimental implementation, where the game is symmetric, and players cannot condition their behavior on specific roles. For simplicity, we assume each agent chooses between two effort levels,  $A=\left\{ L,H\right\} $ (the results are similar
with more effort levels). An agent selecting low effort $L$ receives a payoff of 1. An agent selecting high effort ($H$) earns a payoff of $2 - c$ if all opponents also choose $H$, and $1 - c$ otherwise, where $c \in (0,1)$ represents the cost of high effort. The game has two strict equilibria: the \textit{safe equilibrium} $\boldsymbol{L}$ and the \textit{efficient equilibrium} $\boldsymbol{H}$. Let $p\in[0,1]$ denote the share of agents playing action $a$ in the popualtion.

When generalizing the sampling dynamics to games with more than two players, different assumptions can be made about what each agent observes. We focus on two alternative assumptions on what each agent observes in each element of her sample:
\begin{enumerate}
\item the minimum effort level in a random round of the $N$-player game. This observation structure fits the best the typical feedback in experimental implementations of the minimum effort games (see, e.g., \citealp{vanHuyck1990tacit,ramos2023road}).
\item The action of a randomly chosen opponent.
\end{enumerate}

\subsubsection{Observation of Minimum Efforts}
Our first result characterizes the asymptotic stability of pure equilibria. The safe equilibrium is always asymptotically stable, while the efficient equilibrium becomes unstable iff the truncated expectation of the sample size is sufficiently large.
\begin{prop}
\label{prop:The-safe-equilibrium}State $L$
is asymptotically stable for all $c$-s. State $H$is:
\begin{enumerate}
\item asymptotically stable if $\mathbb{E}_{\leq\frac{1}{1-c}}\left(\theta\right)<\frac{1}{N}$,
and
\item unstable if $\mathbb{E}_{<\frac{1}{1-c}}\left(\theta\right)>\frac{1}{N}$.
\end{enumerate}
\end{prop}

\begin{proof}
 Consider a perturbed state $\left(1-\epsilon\right)$ near $L,$ where $\epsilon<<1$ represents the small share of agents playing $H$. A new agent with sample size $k$ will only play $H$ if they observe $H$ as the minimum effort in their sample, but the probability of this is negligible: $O(k\cdot\epsilon^N)<\epsilon$. Thus, $L$ is asymptotically stable.

Now consider a perturbed state $\epsilon<<1$ near $H$. A new agent with sample size $k$ plays $H$ (resp. $L$) when they observe a rare minimum effort level $L$ once in their sample if $\left(k-1\right)\left(1-c\right)>c\Leftrightarrow k>\frac{1}{1-c}$ (resp. $k<\frac{1}{1-c}$). The probability of observing $L$ is approximately $N\cdot k\cdot\epsilon$. Following similar reasoning to Proposition \ref{prop:pure-equilbiria-satbility}, we conclude:
\begin{enumerate}
    \item converging to everyone playing $H$, so $H$ is asymptotically stable if
    $N\cdot\mathbb{E}_{\leq\frac{1}{1-c}}\left(\theta\right)<1.$    
    \item the share of agents playing $L$ increases, making $H$ unstable if  $N\cdot\mathbb{E}_{<\frac{1}{1-c}}\left(\theta\right)>1.$   \qedhere
\end{enumerate}
\end{proof}
Comparative statics for the stability of the efficient equilibrium align with experimental findings: the set of distributions for which the efficient equilibrium is stable decreases with both the cost of effort $c$ and the number of players $N$. (Numeric analysis also suggests that similar trends apply to the size of the efficient equilibrium's basin of attraction.)

\begin{rem}
In typical experiments of this game, there are 7 (rather than 2) levels of
effort, and a player's payoff is equal to the minimal effort
level chosen by any of the players minus $c$ times her own effort.
Simple adaptations to the proof of Proposition \ref{prop:The-safe-equilibrium}
show that the same condition for the asymptotic stability of the efficient equilibrium holds for any non-safe equilibrium. That is, if $\mathbb{E}_{<\frac{1}{1-c}}\left(\theta\right)>\frac{1}{N-1}$,
then only the safe equilibrium is asymptotically stable, while if
$\mathbb{E}_{\leq\frac{1}{1-c}}\left(\theta\right)<\frac{1}{N-1},$
then all pure equilibria are asymptotically stable.
\end{rem}
Next we show that heterogeneity induces stable miscoordination also with $N>2$ players. Specifically, we generalize Theorem \ref{thm:locally-stableinterior_ASD}, and show that if $k$ is not too large, one can always add players with sufficiently large sample sizes, and obtain an environment with stable miscoordination.

\begin{prop}[Adaptation of Theorem \ref{thm:locally-stableinterior_ASD}]\label{prop-hetro-stable-interoir-minimum}
For any $k<\frac{1}{1-c}$, there exists a minimum-effort environment with some agents having sample size $k$ and the others with sufficiently large samples, in which an asymptotically stable interior state exists.
\end{prop}
\begin{proof}
Let $p^{NE}\in\left(0,1\right)$ be the symmetric interior Nash equilibrium
of the minimum effort coordination game. Fix a sufficiently small
$\epsilon>0$. Let $w^{\theta}\left(p\right)$ be the probability
of a new agent playing action $L$ when the
new agent's sample size is distributed according to $\theta$, and
when the share of agents playing $L$ is $p$. 
An agent with sample size $k$ would play action $L$ if her sample
includes at least one observation of $L$. The probability of this is
\[
w^{k}\left(p\right)=1-\left(1-p\right)^{k\left(N-1\right)}=k\left(N-1\right)p-\left(\begin{array}{c}
k\left(N-1\right)\\
2
\end{array}\right)p^{2}+O\left(p^{3}\right).
\]
Fix a sufficiently small $\epsilon>0$. Let $\hat{p}\in\left(0,p^{NE}-\epsilon\right)$
be sufficiently small such that the term $O\left(p^{3}\right)<\epsilon$
is negligible for any $p<\hat{p}$. Let $\alpha\in\left(0,1\right)$
be such that: (1) $\alpha k\left(N-1\right)>1$ and (2) $\alpha\left(k\left(N-1\right)-\left(\begin{array}{c}
k\left(N-1\right)\\
2
\end{array}\right)p\right)<1-2\epsilon$. This implies that $w^{k}\left(p\right)>\frac{p}{\alpha}$ in a right
neighborhood of zero, and $w^{k}\left(p\right)<\frac{p}{\alpha}$
in a left neighborhood of $\hat{p}$. Observe that $w^{\bar{k}}\left(p\right)<\epsilon$
for a sufficiently large $\bar{k}$. This implies that $w^{k\alpha\bar{k}}\left(p\right)>p$
in a right neighborhood of zero, and $w^{k\alpha\bar{k}}\left(p\right)<p$
in a left neighborhood of $\hat{p}$. This, in turn, implies that
there exists a symmetric stationary state $\tilde{p}\in\left(0,\hat{p}\right)$
that satisfies (1) $w^{k\alpha\bar{k}}\left(\tilde{p}\right)=\tilde{p}$,
(2) $w^{k\alpha\bar{k}}\left(p\right)>p$ in a left neighborhood of
$\tilde{p}$, and (3) $w^{k\alpha\bar{k}}\left(p\right)<p$ in a right
neighborhood of $\tilde{p}$. Due to Part (3) of Fact \ref{fact-w-stable} $\tilde{p}$ is asymptotically stable.\qedhere
\end{proof}

\subsubsection{Observation of Actions}
Next, we characterize the asymptotic stability of pure equilibria when new agents observe actions instead of minimum efforts. 
\begin{prop}
\begin{enumerate}
\item The safe equilibrium $L$ is:
\begin{enumerate}
    \item  asymptotically stable if $\mathbb{E}_{\leq\left(\sqrt[\left(N-1\right)]{\frac{1}{1-c}}\right)}\left(\theta\right)<1,$ and
    \item unstable if $\mathbb{E}_{<\left(\sqrt[\left(N-1\right)]{\frac{1}{1-c}}\right)}\left(\theta\right)>1.$
\end{enumerate}
\item The efficient equilibrium
$H$is:
\begin{enumerate}
\item asymptotically stable if $\mathbb{E}_{\leq\left(\frac{1}{1-\sqrt[\left(N-1\right)]{c}}\right)}\left(\theta\right)<1,$ and
\item unstable if $\mathbb{E}_{<\left(\frac{1}{1-\sqrt[\left(N-1\right)]{c}}\right)}\left(\theta\right)>1.$
\end{enumerate}
\end{enumerate}
\end{prop}

\begin{proof}
 \begin{enumerate}
     \item Consider a perturbed state $\left(1-\epsilon\right)$ near $L,$ where $\epsilon<<1$ represents the small share of agents playing $H$. A new agent with sample size $k$ who observes the rare action $H$ once in their sample estimates the probability that the minimum effort of $N-1$ random opponents is $H$ as $\frac{1}{k^{N-1}}$. This means that the agent will play $H$ (resp., $L$) if $\left(k^{(N-1)}-1\right)\left(1-c\right)>c\Leftrightarrow k^{(N-1)}>\frac{1}{1-c}$ (resp. $k^{(N-1)}<\frac{1}{1-c}$). The probability of observing $L$ is approximately $k\cdot\epsilon$. Following similar reasoning to Proposition \ref{prop:pure-equilbiria-satbility}, we obtain conditions (a) and (b).

\item Consider a perturbed state $\epsilon<<1$ near $H$. A new agent with sample size $k$ who observes the rare action $L$ once estimates the probability that the minimum effort of $N-1$ random opponents is $L$ as 
$1-\left(\frac{k-1}{k}\right)^{N-1}.$ 
This means that the agent will play $H$ (resp., $L$) if 
$\left(1-\left(\frac{k-1}{k}\right)^{N-1}\right)c>\left(1-c\right)\left(\frac{k-1}{k}\right)^{N-1} \Leftrightarrow k>\frac{1}{1-\sqrt[\left(N-1\right)]{c}}$ 
(resp. $k<\frac{1}{1-\sqrt[\left(N-1\right)]{c}}$). 
The probability of observing $L$ is approximately $k\cdot \epsilon$. Following similar reasoning to Proposition \ref{prop:pure-equilbiria-satbility},  we obtain conditions (a) and (b).\qedhere
 \end{enumerate} 

\end{proof}

Finally, we show that heterogeneity induces stable miscoordination also with $N>2$ players. Specifically, we generalize Theorem \ref{thm:locally-stableinterior_ASD}, and show that if $k$ is not too large, one can add players with sufficiently large sample sizes, and obtain an environment with stable miscoordination. The proof, which is analogous to Proposition \ref{prop-hetro-stable-interoir-minimum}, is omitted for brevity.

\begin{prop}[Adaptation of Theorem \ref{thm:locally-stableinterior_ASD}]\label{prop-hetro-stable-interoir-minimum}
For any $1<k<\frac{1}{1-\sqrt[\left(N-1\right)]{c}}$, there exists a minimum-effort environment with some agents having sample size $k$ and the others with sufficiently large samples, in which an asymptotically stable interior state exists.
\end{prop}

\subsection{Logit Dynamics}\label{sec-logit}\label{sec:logit}
In this appendix, we show that our result of heterogeneity inducing stable miscoordination applies to logit dynamics. Specifically, we find that (1) standard logit dynamics with uniform noise levels require implausibly high noise to achieve stable miscoordination, while (2) a variant with heterogeneous noise levels achieves stable miscoordination with much lower noise. This suggests our insight into the role of heterogeneity in noise may be relevant across different dynamics, not just sampling dynamics.

\paragraph{Standard (Homogeneous) Logit Dynamics }

Logit dynamics (introduced in \citealp{fudenberg1995consistency}; see \citealp[Section 6.2.3]{sandholm2010population}
for a textbook exposition, 
and see \citealp{nax2022deep} for a recent application) are characterized by a single parameter
$\eta_{i}$ that describes the \emph{noise level }for each population
$i$. If player $i$ plays action $a_{i}$, she gets a payoff
of $p_{j}\cdot u_{j}$. If she plays action $b_{i}$ she gets
a payoff of $\left(1-p_{j}\right)\cdot1$. Logit dynamics assume
that the probability of revising agents playing action $a_{i}$
is proportional to $e^{\frac{\textrm{Payoff of \ensuremath{a_{i}}}}{\eta}}$.
Specifically, logit dynamics are given by

\begin{equation}
w_{i}\left(p_{j}\right)\equiv w_{i}\left(\textbf{p}\right)=\frac{e^{\frac{p_{j}\cdot u_{j}}{\eta_{i}}}}{e^{\frac{1-p_{j}}{\eta_{i}}+}e^{\frac{p_{j}\cdot u_{j}}{\eta_{i}}}}.\label{eq:logit-dyanmics}
\end{equation}

Trivially, logit dynamics can induce substantial miscoordination by
having high values of noise. The interesting question is whether stable
miscoordination can be supported by a low level of noise. Our numerical
analysis suggests that the answer is negative. In what follows, we demonstrate that this is indeed the case. For example, when we revisit the two examples of Figure \ref{fig:fig3-local-stablity} ($u_{1}=u_{2}=2.5$
and $u_{1}=\frac{1}{u_{2}}=5$), then the minimal level of noise that
is required to sustain an asymptotically stable interior state in which
each action is played with a probability of at least 10\% is $\eta=1$
(see the left panel of Figure \ref{fig4-logit-hetro} for an illustration
of the case of $u_{1}=u_{2}=2.5$). Such a high level of noise implies
that 27\% of the revising agents make the obvious mistake of playing
$a_{i}$ when facing a population in which almost everyone
plays $b_{j}$; by contrast, this obvious mistake is never made under
action-sampling dynamics. Moreover, the average expected payoff obtained
by revising agents who follow logit dynamics against an opponent
population in which the share of agents playing action $a_{i}$ is
distributed uniformly is 85\% (resp., 71\%) of the maximal payoff
that can be obtained by payoff-maximizing agents in the first (resp.,
second) environment with $u_{1}=u_{2}=2.5$ (resp., $u_{1}=\frac{1}{u_{2}}=5$).
By contrast, this average expected payoff is 98\% (resp., 95\%) of
the maximal payoff under the sampling dynamics. Thus, stable cooperation
can be supported by standard (homogeneous) logit dynamics only when the agents have high levels of noise. 

\paragraph{Heterogeneous Logit Dynamics }

\begin{figure}[h]
\caption{Supporting Stable Coordination with Heterogeneous Logit Dynamics\label{fig4-logit-hetro}}
\includegraphics[scale=0.59]{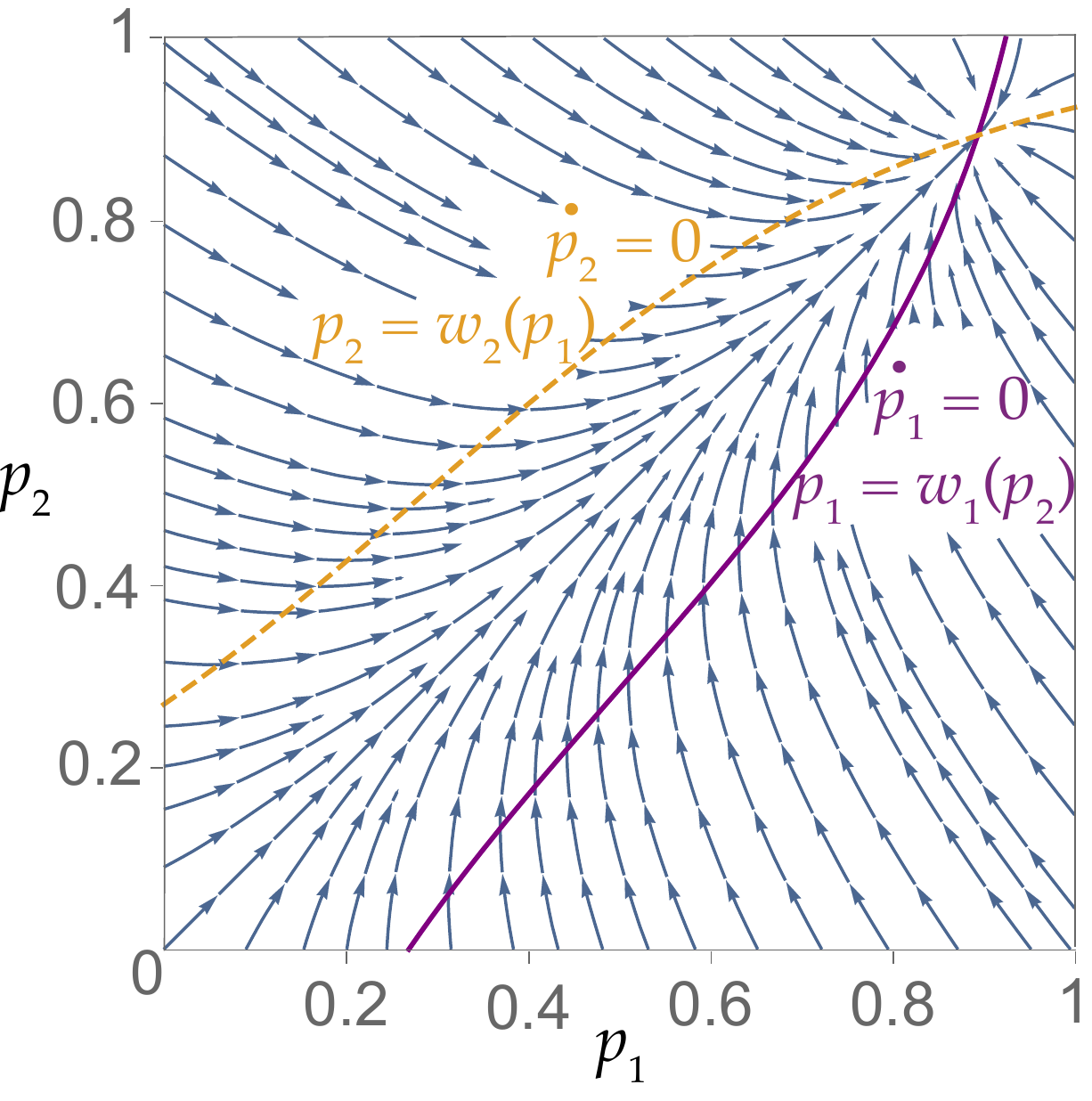}~~~~~~\includegraphics[scale=0.59]{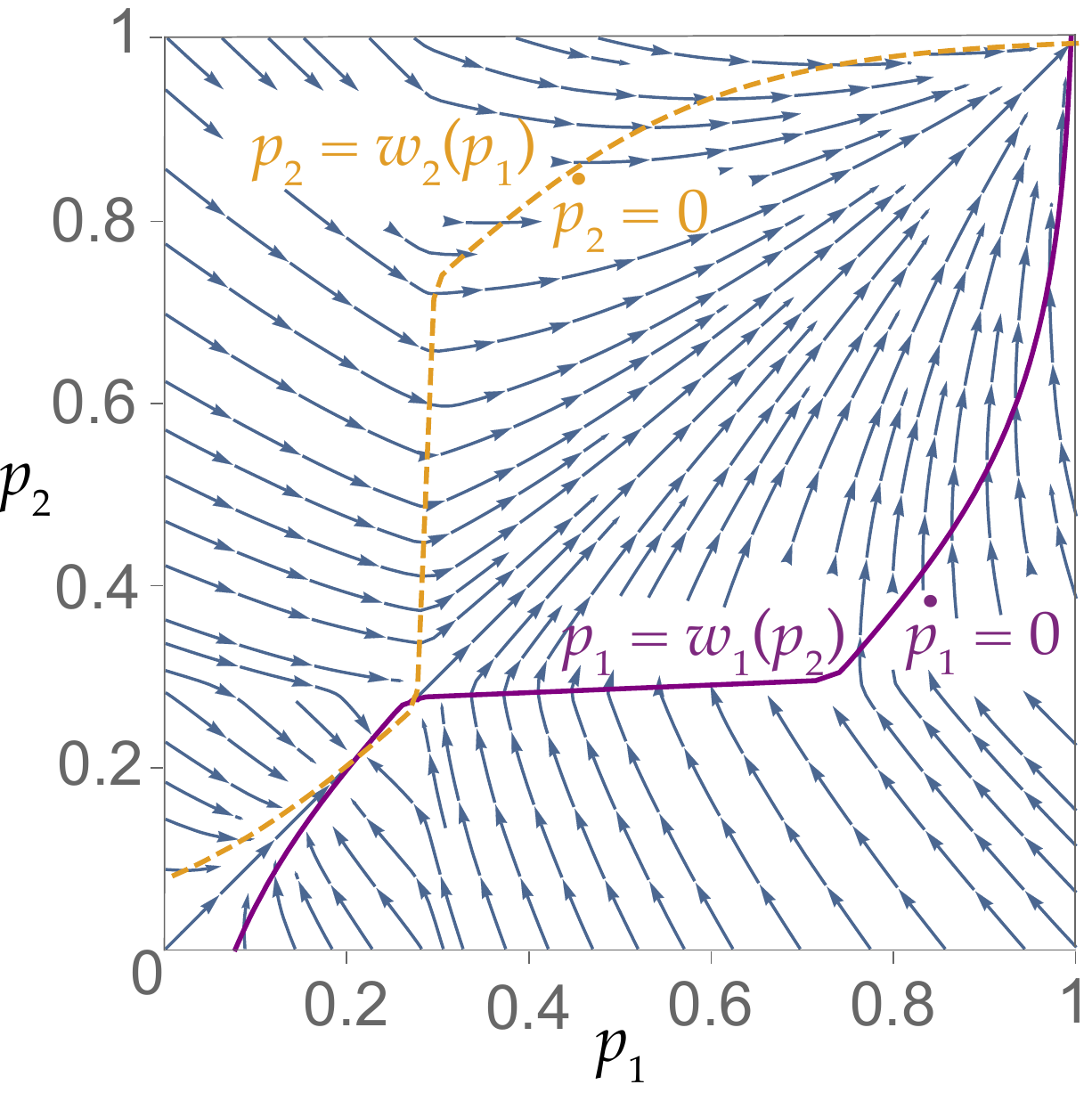}
\small{
The figure revisits the symmetric game presented in the left panel of Figure
\ref{fig:fig3-local-stablity} in which $u_{1}=u_{2}=2.5$. The left
panel shows the phase plot of the minimal homogeneous level of noise, $\eta_{i}=1$, that sustains an asymptotically stable state
in which each action is played with a probability of at least 10\%.
The right panel shows the phase plot of a heterogeneous variant of
logit dynamics in which 55\% of the the agents in each population
have a moderate level of noise $\eta_{i}=0.55$ and 45\% have a small
level of noise $\eta_{i}=0.01$.}
\end{figure}
Next, consider a variant of logit dynamics in which there is heterogeneity
in the level of noise for agents in each population. Specifically, in
a population in which there are $n$ groups, the size of the $l$-th
group is $\mu_{i}^{l},$ and its members have a noise level of $\eta_{i}^{l},$
the heterogeneous logit dynamics are given by
\begin{equation}
w_{i}\left(p_{j}\right)\equiv w_{i}\left(\textbf{p}\right)=\sum_{l}\mu_{i}^{l}\cdot\frac{e^{\frac{p_{j}\cdot u_{j}}{\eta_{i}^{l}}}}{e^{\frac{1-p_{j}}{\eta_{i}^{l}}+}e^{\frac{p_{j}\cdot u_{j}}{\eta_{i}^{l}}}}.\label{eq:logit-dyanmics-1}
\end{equation}
The numerical calculations demonstrate that heterogeneous noise levels
can induce asymptotically stable miscoordination with relatively low
levels of noise. Specifically, in both of the above examples ($u_{1}=u_{2}=2.5$,
which is illustrated in the right panel of Figure \ref{fig4-logit-hetro},
and $u_{1}=\frac{1}{u_{2}}=5$), populations in which 55\% of the agents have a moderate level of noise (i.e., $\eta=0.55$) and 45\% have
a small level of noise (i.e., $\eta=0.01$) induce asymptotically stable states
with miscoordination ($\left(0.21,0.21\right)$ in the right panel
of Figure \ref{fig4-logit-hetro}). Given these heterogeneous levels
of noise, only 8\% of the agents make the mistake of playing action
$a_{i}$ when facing a population in which everyone plays
$a_{j}$, and the average expected payoff obtained by playing against opponent
populations in which the share of agents playing action $a_{i}$ is
distributed uniformly is 96\% (resp., 89\%) of the maximal payoff
that can be obtained by payoff-maximizing revising agents in the environment
with $u_{1}=u_{2}=2.5$ (resp., $u_{1}=\frac{1}{u_{2}}=5$).

 \end{document}